\documentclass[article]{IEEEtran}

\usepackage{hyperref}  
\usepackage{soul}
\usepackage{subfig}
\usepackage{graphicx}
\usepackage{amsmath}
\usepackage{amsthm}
\usepackage[capitalise]{cleveref}
\usepackage{longtable}
\usepackage{enumitem}
\usepackage{amssymb}
\usepackage{courier}
\usepackage{amsfonts}
\usepackage{enumitem}
\usepackage{graphics}
\usepackage{graphicx}
\usepackage{lipsum}
\usepackage{venndiagram}
\usepackage[font={small},belowskip=-50pt]{caption}
\usepackage[linesnumbered,ruled,vlined]{algorithm2e}
\usepackage{algpseudocode}
\usepackage{color}
\usepackage{array}
\usepackage{cite}
\usepackage{url}
\usepackage{parskip}
\usepackage{lettrine}

\usepackage{multirow}
\usetikzlibrary{positioning}
\usepackage{floatrow}
\newcommand{\mf}{\mathbf}
\newcommand{\be}{\begin{eqnarray*}}
\newcommand{\ee}{\end{eqnarray*}}
\newcommand{\ben}{\begin{eqnarray}}
\newcommand{\een}{\end{eqnarray}}

\newcommand{\Snst}{\mathbb{S}_{\text{nest}}}
\newcommand{\Sula}{\mathbb{S}_{\text{ula}}}

\newcommand{\blds}{\boldsymbol}
\theoremstyle{definition}
\newtheorem{definition}{Definition}[section]

\newcommand{\E}{{\mathbb E}}

 \newcommand{\Ryov}{\mf{R}_{\text{av}}}
\newcommand{\Ryhat}{\mf{\widehat{R}}_{\mf{y}}}
\newcommand{\Ry}{\mf{R}_{\mf{y}}}
\newcommand{\Ey}{\mf{E}_{\mf{y}}}
\newcommand{\EL}{\mf{E}_{L}}
\newcommand{\Lamth}{\mf{\Lambda}(\theta)}
\newtheorem{lem}{Lemma}
\newtheorem{thm}{Theorem}
\Crefname{thm}{Theorem}{Theorems}
\crefname{thm}{Theorem}{Theorems}

\newtheorem{cor}{Corollary}
\newtheorem{prop}{Proposition}

\makeatletter
\newcommand*{\rom}[1]{\expandafter\@slowromancap\romannumeral #1@}
\makeatother

\title{Super-resolution with Sparse Arrays: A Non-\\ Asymptotic Analysis of Spatio-temporal Trade-offs}
\author{Pulak Sarangi, Mehmet Can H\"uc\"umeno\u{g}lu, Robin Rajam\"{a}ki, and Piya Pal}

\begin{document}

\maketitle
\begin{abstract}
Sparse arrays have emerged as a popular alternative to the conventional uniform linear array (ULA) due to the enhanced degrees of freedom (DOF) and superior resolution offered by them. In the passive setting, these advantages are realized by leveraging correlation between the received signals at different sensors. This has led to the belief that sparse arrays require a large number of temporal measurements to reliably estimate parameters of interest from these correlations, and therefore they may not be preferred in the sample-starved regime. In this paper, we debunk this myth by performing a rigorous non-asymptotic analysis of the Coarray ESPRIT algorithm. This seemingly counter-intuitive result is a consequence of the scaling of the singular value of the coarray manifold, which compensates for the potentially large covariance estimation error in the limited snapshot regime. Specifically, we show that for a nested array operating in the regime of fewer sources than sensors ($S=O(1)$), it is possible to bound the matching distance error between the estimated and true directions of arrival (DOAs) by an arbitrarily small quantity ($\epsilon$) with high probability, provided (i) the number of temporal snapshots ($L$) scales only logarithmically with the number of sensors ($P$), i.e. $L=\Omega(\ln(P)/\epsilon^2)$, and (ii) a suitable separation condition is satisfied. Our results also formally prove the well-known empirical resolution benefits of sparse arrays, by establishing that the minimum separation between sources can be $\Omega(1/P^2)$, as opposed to separation $\Omega(1/P)$ required by a ULA with the same number of sensors. In addition to the array geometry, our sample complexity expression reveals the dependence on other key model parameters such as Signal to Noise Ratio (SNR) and the dynamic range of the source powers. This enables us to establish the superior noise-resilience of nested arrays both theoretically and empirically. \footnote{This work was supported by Grants ONR N00014-19-1-2256, ONR N00014-19-1-2227, NSF 2124929, and NSF CAREER ECCS 1700506.}
\end{abstract}
\vspace{-0.2cm}
\vspace{-0.5cm}
\section{Introduction}
The problem of source localization arises in different contexts ranging from target detection in sonar and radar, hybrid mmWave channel estimation, and DOA estimation in array signal processing \cite{pal2010nested,vaidyanathan2011sparsesamplers,haghighatshoar2018low}. Traditionally, these applications consider ULAs, which are known to resolve up to $S=O(P)$ sources with $P$ sensors. However, deterministic sparse array geometries, such as nested and coprime arrays \cite{pal2010nested,vaidyanathan2011sparsesamplers}, have recently gained significant attention primarily due to two attractive properties. Firstly, sparse arrays are able to identify up to $S=O(P^2)$ uncorrelated sources using only $P$ sensors. Secondly, sparse arrays enjoy a performance gain showcased by lower Cram$\acute{\text{e}}$r-Rao bound and higher angular resolution \cite{abramovich1998positive,wang2017,koochakzadeh2016cramerrao,liu2017cramerrao,shahsavari2021fundamental}. Both of these properties can be attributed to the enhanced spatial DOF enabled by the so-called difference coarray, which {can be as large as $\Theta(P^2)$.}

The enhanced DOF of the coarray are realized by computing temporal correlations between the spatial measurements and constructing an augmented covariance matrix called the ``coarray covariance matrix", whose size is determined by the size of the difference coarray. Following the construction of the coarray covariance matrix, it is possible to fully harness the power of the difference coarray and identify the unknown source directions using classical subspace techniques, such as MUSIC, ESPRIT or the matrix pencil method \cite{schmidt1986multiple,roy1989esprit,hua1990matrix}. Despite the success of coarray-based algorithms, a common belief is that they require a large number of temporal snapshots to fully utilize the number of DOFs provided by the coarray. The root of this belief mainly lies in the inadequacy of existing performance analyses, which are primarily based on characterizing the asymptotic Mean Squared Error (MSE) of the Coarray MUSIC \cite{wang2017} and Coarray ESPRIT algorithms \cite{steinwandt2017performance}. In particular, such asymptotic results primarily rely on the first-order perturbation analysis framework proposed in \cite{Vaccaro93}, which leaves two key questions unanswered regarding the performance of coarray algorithms. Firstly, the perturbation framework fails to theoretically explain the improvement in resolution offered by sparse arrays over the ULA---a phenomenon that has been extensively observed in numerical experiments {\cite{wang2017,sun20214d}}. Secondly, the analysis does not adequately reveal the dependence of temporal snapshots on key model parameters such as the array geometry, number of sensors, SNR and dynamic range of the source powers. 

The aforementioned shortcomings are partially addressed in \cite{NehoraiEldar}, which adapts recent advances in the theory of super-resolution \cite{candes2013super,candes2014towards} to the coarray setting. The analysis, which is based on Total-Variational norm minimization, is indeed non-asymptotic. However, it is possible to show that the snapshot requirement in this setting scales quadratically (rather than linearly) with the number of sensors $P$, which is undesirable. In a parallel line of work using a grid-based model, we recently showed that $\Omega(P^2)$ snapshots are sufficient for ensuring exact support recovery with high probability even for closely-spaced sources, where the smallest source separation scales as $\Omega(1/P^2)$ \cite{qiao2019guaranteed}. Although the analysis is applicable for scenarios where $S>P$ (more sources than sensors), the sample complexity $\Omega(P^2)$ is still conservative {when $S\leq P$}. 
{In \cite{zhou2018direction}, an atomic norm formulation is adopted to exploit the Toeplitz structure of the coarray covariance matrix. The analysis provides a characterization of the covariance matrix estimation error, but not of the sample complexity required to achieve a desired DOA estimation error, which is often the main quantity of interest.} 
Indeed, common folklore suggests that the benefits of sparse arrays necessarily come at the cost of a large number of snapshots, since the coarray covariance matrix, which typically needs to be estimated, is of size $\Theta(P^2)$. Hence, one might be tempted to \emph{falsely} conclude that sparse arrays are at a disadvantage compared to ULAs. In this paper, our goal is to dispel this belief by providing new non-asymptotic results on the performance of Coarray ESPRIT with a focus on nested arrays {in the regime $S\leq P$.} Our analysis is motivated by contemporary applications such as autonomous sensing and mmWave channel estimation \cite{sun20214d,haghighatshoar2018low}, where identifying more sources than sensors may not be necessary, and the number snapshots may be restricted either due to coherent multipaths or a rapidly varying environment.

While subspace-based algorithms have been around for several decades and actively used in practice, performance guarantees characterizing their precise resolution limit were obtained only recently \cite{WeilinEsprit,li2021stable,liao2016music,liao2015music,moitra2015super}. This analysis has also been extended to multi-snapshot setting in \cite{WeilinMulti}. The key factor enabling these guarantees is the characterization of the smallest singular value of Vandermonde matrices \cite{moitra2015super}. However, all the aforementioned results are only applicable to the ULA. Furthermore, \emph{no statistical assumptions} are made on the source signals, and hence, the coarray perspective is missing. The key difference between deterministic and random sources is that in the latter case, the perturbation to the subspace of interest is a consequence of both noise as well as finite-snapshot covariance estimation error. Therefore, extending the analysis in \cite{WeilinEsprit,WeilinMulti} to the stochastic case requires non-trivial modifications.

\textbf{Contributions:}
Our first main contribution is to probabilistically characterize the coarray covariance matrix estimation error due to finite snapshots. Our second main contribution is a non-asymptotic performance analysis for the Coarray ESPRIT algorithm in terms of the matching distance error metric. Specifically, we characterize the number of temporal snapshots (sample complexity) required to bound the matching distance error by a specified parameter. To the best of our knowledge, our sample complexity expression (in terms of snapshots) is the first to explicitly bring out the dependence on key model parameters such as the array geometry, SNR and dynamic range of the source powers. Furthermore, we establish that it is possible to bound the matching distance error with an arbitrarily small quantity for both the nested array and ULA, using the (order-wise) same number of snapshots $L = \Omega(\ln P)$. However, a nested array can achieve this in a much smaller separation regime $\Delta_{\min}=\Omega(1/P^2)$ compared to the ULA, for which $\Delta_{\min}=\Omega(1/P)$. Our analysis dispels the widely-held belief that sparse arrays require significantly more snapshots compared to ULAs when the number of sources is less than the number of sensors, and at the same time establishes the superior resolution capabilities of nested arrays. In addition to advancing the theoretical understanding, {this analysis could also serve as a guiding principle for practitioners to determine suitable operating conditions.}
\textit{Notations:} 
Symbol $\odot$ represents the Khatri-Rao (columnwise Kronecker) product, whereas $\Vert \cdot \Vert_2$ and $\Vert \cdot \Vert_F$ denote the spectral and Frobenius norm of a matrix. Moreover, $\sigma_i(\mathbf{A})$ is the $i$-th largest singular value of $\mathbf{A}$. For a set real numbers $\{p_1,p_2, \hdots,p_K \}$, $p_{\min}$ and $p_{\max}$ denote the minimum and maximum numbers in the set, respectively. The symbol $\mathbb{T}:= [0,1)$ denotes the torus. For a sub-Gaussian random variable $X$, $\Vert X \Vert_{\psi_2}$ denotes its sub-Gaussian norm defined as $\Vert X \Vert_{\psi_2}:=\inf \{ t>0 \ \vert \ \mathbb{E} [\exp{X^2/t^2}] \leq 2\}$.

\section{Background on Sparse Arrays}\label{sec:background}
Consider a sparse linear array (SLA) with $P$ sensors located at $\{d_p\lambda/2\}_{p=1}^{P}$, where $\lambda$ is the wavelength of the incoming far-field narrow-band source signals and $d_p$ belongs to an integer set $\mathbb{S}$ ($\vert \mathbb{S}\vert =P$). Suppose $S$ sources with distinct DOAs $\blds{\theta}=\{\theta_1,\theta_2,\cdots,\theta_{S}\}$  impinge on the array where $\theta_i \in (-\pi/2, \pi/2]$  for $i= 1, \hdots, S$. The signal received at the $P$ sensors at time instance $t$ is given by: 
\begin{equation}\label{measurementmodel}
\mathbf{y}(t) = \mathbf{A}_{\mathbb{S}}(\blds{\theta})\mathbf{x}(t) + \mathbf{n}(t),\quad t = 1,\hdots,L.    
\end{equation}
The matrix $\mathbf{A}_{\mathbb{S}}(\blds{\theta}) = [\mathbf{a}_{\mathbb{S}}(\theta_1), \mathbf{a}_{\mathbb{S}}(\theta_2),\hdots,\mathbf{a}_{\mathbb{S}}(\theta_S)] \in \mathbb{C}^{P \times S}$ is the array manifold matrix where: $\mathbf{a}_{\mathbb{S}}(\theta_i) = [e^{j \pi d_1 \sin(\theta_i)},e^{j \pi d_2 \sin(\theta_i)} \ \hdots \ e^{j \pi d_P \sin(\theta_i)}]^{\top},$ represents the steering vector corresponding to the direction $\theta_i$, $L$ denotes the total number of temporal snapshots, $\mathbf{x}(t) \in \mathbb{C}^{S}$ is the $t^{\text{th}}$ temporal snapshot of the source signal vector and $\mathbf{n}(t) \in \mathbb{C}^{P}$ is an additive noise term. We define the normalized spatial frequencies (which we refer to as normalized DOAs) as $\omega_i=\sin(\theta_i)/2$. Throughout this paper, we make the following statistical assumptions on the source signals and noise:
\begin{itemize}[leftmargin=0.5mm]
\item[] \textbf{[A1]} \textbf{Uncorrelated Gaussian Sources}: The source signals $\mf{x}(t)$ are assumed to be uncorrelated white circularly symmetric Gaussian $\mathcal{CN}(\mathbf{0},\mathbf{P})$ where $\mathbf{P} = \text{diag}(p_1,p_2,\hdots,p_S)$ represents a diagonal covariance matrix of {source powers}. 
\item[] \textbf{[A2]} \textbf{Gaussian Noise}: The noise $\mf{n}(t)$ follows a zero-mean circularly symmetric complex Gaussian distribution $\mf{n}(t)\sim\mathcal{CN}(\mathbf{0},\sigma^2 \mathbf{I})$, and is uncorrelated with $\mf{x}(t)$.
\end{itemize}

Under assumptions [\textbf{A1-A2}], the measurements follow $\mathbf{y}(t)\sim\mathcal{CN}(\mathbf{0},\mathbf{R}_{y})$, where $\mathbf{R}_{y}$ is given by:
\ben
\smash{\mathbf{R}_{y} = \mathbf{A}_{\mathbb{S}}(\blds{\theta})\mf{P}\mf{A}^H_{\mathbb{S}}(\blds{\theta}) + \sigma^2\mathbf{I}_P \in \mathbb{C}^{P \times P}.} \label{eqn:Ry_gt}
\een
By vectorizing $\mathbf{R}_y$, we obtain the ``virtual measurements": $\mathbf{r}_y= (\mathbf{A}_{\mathbb{S}}(\blds{\theta})^{*}\odot \mathbf{A}_{\mathbb{S}}(\blds{\theta}))\mathbf{p} + \sigma^2\mathbf{i},$
 where $ \mathbf{i}=\text{vec}(\mathbf{I}_P) $ and $\mathbf{p}=  [p_1,\ldots,p_S]^T$. The matrix  $\mathbf{A}_{\mathbb{S}}(\blds{\theta})^{*}\odot \mathbf{A}_{\mathbb{S}}(\blds{\theta})$ can be viewed as a ``virtual array" with sensor locations given by the difference set of the SLA.
\begin{definition}[Difference Set]
 Given a SLA $\mathbb{S}=\{d_1,d_2,\cdots,d_P\}$, its difference set  $\mathbb{D}_{\mathbb{S}}$ is defined as: $\smash{\mathbb{D}_{\mathbb{S}} =\{  d_m-d_n \vert d_m,d_n \in \mathbb{S}\}}$.
\end{definition}
The difference set $\mathbb{D}_{\mathbb{S}}$ of $\mathbb{S}$ is also called its virtual difference coarray. Let $M_{ca}>0$ be the largest integer such that the set $\mathbb{U}_{\mathbb{S}} := \{0,1, \hdots, M_{ca}\}$ satisfies $\mathbb{U}_{\mathbb{S}} \subseteq \mathbb{D}_{\mathbb{S}}$. This set $\mathbb{U}_{\mathbb{S}}$ denotes the largest contiguous non-negative segment of the difference set and is essentially a ULA with $M_{ca}+1$ sensors. By harnessing the structure of $\mathbb{U}_{\mathbb{S}}$, sparse arrays enjoy enhanced degrees of freedom over the physical SLA. An array is called hole-free if its difference set is a ULA, i.e., $\mathbb{D}_{\mathbb{S}}=\{-M_{ca},\cdots,M_{ca}\}$. We now introduce the notation for a {``generalized nested array"}, which is a special hole-free array.
\begin{definition}[Nested array]
 A generalized nested array $\mathbb{S}^{(N_1,N_2)}$ with $N_1\geq N_2>0$, is defined as: $\mathbb{S}^{(N_1,N_2)}=\{n\}_{n=1}^{N_1}\ \cup\ \{m(N_1+1)\}_{m=1}^{N_2}$.
\end{definition}
It can be shown that any nested array $ \mathbb{S}^{(N_1,N_2)}$ is hole-free, i.e.,  $\mathbb{U}_{\mathbb{S}}=\{0,1,\cdots,M_{\text{ca}}\}$ with $M_{\text{ca}} =N_2(N_1+1)-1$. Furthermore, $\mathbb{S}_{\text{ula}}=\mathbb{S}^{(P-1,1)}$, i.e., choosing $N_1=P-1$ and $N_2=1$, yields a ULA with $P$ sensors.
For a given $P$, if $N_1 = \lceil \frac{P}{2} \rceil$, $N_2 = \lfloor \frac{P}{2}\rfloor $, then $M_{\text{ca}}+1=\lfloor \frac{P}{2}\rfloor(\lceil \frac{P}{2} \rceil+1)$. It can be verified that for $P\geq 3$, we have: 
\ben
P^2/5\leq M_{\text{ca}}+1\leq P^2\label{eqn:M_ca_bound}.
\een
Therefore, $M_{ca} = \Theta(P^2)$ is indeed achievable. 
Next, we introduce an important quantity that is essential for describing correlation-based processing. \begin{definition}[Weight Function]\label{def:weight_func}
Consider a hole-free array $\mathbb{S}$. For every $i \in \mathbb{D}_{\mathbb{S}}$, its weight function is defined as  $\vert\Omega_i\vert$:
$\Omega_i = \{ (m,n) | d_m - d_n = i, 1 \le m,n \le P \}$
where the set $\Omega_i$ essentially captures all pairs $(d_m,d_n)$ of sensor locations that generate the difference of $i=d_m-d_n$.
\end{definition}
Due to symmetry, it can be verified that $\vert\Omega_i\vert = \vert\Omega_{-i}\vert$. Next, we review the widely-used ``redundancy averaging" technique used for correlation-domain processing.
Following \cite{pal2010nested,Chunlin,wang2017}, the virtual ULA measurements are given by:
\ben \label{eqn:v_measurement}
\smash{\mathbf{t}=\mf{F}_{\text{av}}\mf{r}_y,}
\een
where $\mf{t}=[t_{-M_{\text{ca}}},\cdots,t_{-1},t_0,t_{1},\cdots,t_{M_{\text{ca}}}]^{\top}$ and $\mathbf{F}_{\text{av}}$ is the redundancy averaging matrix given by:
\ben
[\mathbf{F}_{\text{av}}]_{i+M_{ca}+1,m + P(n-1)} = \begin{cases} \frac{1}{|\Omega_i|} & \text{If }  d_m - d_n  = i\\
0 & \text{Otherwise},
\end{cases}
\label{eqn:F_avg}
\een 
with $-M_{ca}\leq i \leq M_{ca}$ and $1\leq m,n \leq P$. The element $t_i$ is obtained by averaging all entries $[\mathbf{R}_\mf{y}]_{m,n}$ whose indices $(m,n)$ generate a difference of $i$, i.e., $d_m-d_n=i$. Define a Toeplitz operator $\mathcal{T}_{M_{ca}}:\mathbb{C}^{2M_{ca}+1}\rightarrow\mathbb{C}^{M_{ca}+1\times M_{ca}+1}$ as: $\smash{[\mathcal{T}_{M_{ca}}(\mf{z})]_{m,n}=z_{M_{ca}+1+m-n}, 1\leq m,n \leq M_{ca}+1}$. 
If the vector $\mathbf{z} \in \mathbb{C}^{2M_{ca}+1}$ is conjugate symmetric, i.e., $z_{M_{ca}+1+i}=z_{M_{ca}+1-i}^{*},i=0,1,\hdots,M_{ca}$, then $\mathcal{T}_{M_{ca}}(\mathbf{z})$ is a Hermitian matrix.
Using the virtual measurement $\mf{t}$, an augmented virtual co-array covariance matrix $\mf{T}_{ca}\in \mathbb{C}^{(M_{ca}+1) \times (M_{ca}+1)}$ is constructed as follows:
\begin{align}\label{eqn:ToeplitzTca}
 \mf{T}_{ca}:=\mathcal{T}_{M_{ca}}(\mathbf{t})= \mf{A}_{\mathbb{U}_{\mathbb{S}}}(\blds{\theta})\mf{P}\mf{A}_{\mathbb{U}_{\mathbb{S}}}(\blds{\theta})^H+\sigma^2 \mf{I}_{M_{ca}+1}.
\end{align}
Once this virtual coarray covariance matrix has been obtained, any subspace-based algorithm \cite{schmidt1986multiple,roy1989esprit} applied to $\mathbf{T}_{ca}$ can exactly recover the source DOAs provided $M_{ca}\geq S$. Hence, this also reveals that by efficiently designing sparse arrays, we can resolve up to $\Theta(P^2)$ sources with only $P$ sensors. In the next section, we describe how the correlation processing is modified in the finite snapshot setting. 
\vspace{-0.3cm}
\subsection{Finite-Snapshot Coarray Covariance Estimation}
Let $\mathbf{\widehat{R}}_{\mf{y}}$ be the sample covariance matrix given by:
\begin{equation}
\smash{\mathbf{\widehat{R}}_{\mf{y}}:=\frac{1}{L}\sum_{t=1}^{L} \mf{y}(t)\mf{y}(t)^H}. \label{eqn:samp_cov}
\end{equation}
With a finite $L$, all the operations on the true covariance matrix are replaced by operations on the sample covariance matrix. First, we apply the redundancy averaging on $\mathbf{\hat{r}_y}$:
\ben
\mathbf{\hat{t}}:= \mathbf{F}_{\text{av}} \mathbf{\hat{r}_y},\text{ where } \mathbf{\hat{r}_y}:=\text{vec}(\mf{\widehat{R}}_{\mf{y}}). \label{eqn:t_hat}
\een
Here $\mathbf{\hat{t}}=[\hat{t}_{-M_{\text{ca}}},\cdots,\hat{t}_{-1},\hat{t}_0,\hat{t}_{1},\cdots,\hat{t}_{M_{\text{ca}}}]^{\top}$ with
$\hat{t}_i=\frac{1}{\vert \Omega_i\vert}\sum_{d_m-d_n=i}[\widehat{\mf{R}}_\mf{y}]_{m,n}.$
Next, the estimated coarray covariance matrix is obtained by constructing a Toeplitz Hermitian matrix from $\mathbf{\hat{t}}$ as follows: 
\begin{equation}
\smash{\mf{\widehat{T}}_{ca}=\mathcal{T}_{M_{ca}}(\mathbf{\hat{t}})}. \label{eqn:T_hat_est}
\end{equation}
For a hole-free sparse array $\mathbb{S}$, from \eqref{eqn:v_measurement}, the elements of the matrix $\mf{R}_{\mf{y}}$ are given by:
\ben
[\mf{R}_{\mf{y}}]_{m,n}=t_{d_m-d_n}\ 1\leq m,n \leq P. \label{eqn:Ry_element}
\een
Similarly, using the estimated coarray covariance matrix $\mf{\widehat{T}}_{ca}$, we define matrix $\mf{R}_{\text{av}} \in \mathbb{C}^{P \times P}$ as
\ben
[\Ryov]_{m,n}:=\widehat{t}_{d_{m}-d_{n}}, \ 1\leq m,n \leq P.\label{eqn:Ryov_element}
\een
This essentially maps the entries $\hat{t}_i$ into a $P\times P$ matrix with the assignments specified by the difference set of the array $\mathbb{S}$. 
Since the sample covariance matrix $\mathbf{\widehat{R}}_{\mf{y}}$ is imperfect, the estimate $\mf{\widehat{T}}_{ca}$ also incurs an error due to a finite number of snapshots. We denote the covariance estimation error as:
\begin{equation}
    \smash{\mathbf{E}_{L} =  \mathbf{T}_{ca}-\mf{\widehat{T}}_{ca}.} \label{eqn:EL}
\end{equation}
The error in estimating the coarray covariance matrix naturally causes errors in DOA estimation as well. Since subspace based algorithms are typically applied to this estimated covariance matrix $\mf{\widehat{T}}_{ca}$, it becomes crucial to probabilistically characterize the estimation error $\mathbf{E}_{L}$ and how it affects the DOA estimation error. This paper provides such a rigorous theoretical characterization of the DOA estimation error with limited snapshots. 
\subsection{Review of Existing Performance Analysis of Coarray-Based Angle Estimation}
The existing performance analyses for coarray-based algorithms are largely asymptotic in nature. In particular, they rely on the first-order perturbation analysis framework proposed in \cite{Vaccaro93}, which has been used to obtain expressions for the mean square error (MSE) of coarray MUSIC \cite{wang2017}, and coarray ESPRIT \cite{steinwandt2017performance}. Consider the eigen decomposition $\mf{T}_{ca}=\mf{U}\mf{\Gamma}_s\mf{U}+\mf{U}_{\perp}\mf{\Gamma}_n\mf{U}_{\perp}^H $,
where $\mf{U} \in \mathbb{C}^{M_{ca}+1\times S}$ and $\mf{U}_{\perp}\in \mathbb{C}^{M_{ca}+1\times M_{ca}+1-S}$ denote the eigenvectors corresponding to the signal and noise subspaces, respectively. The corresponding perturbed matrices are denoted as $\mf{\widehat{T}}_{ca}=\mf{T}_{ca}+\Delta\mf{T}_{ca}$, $\mf{\widehat{U}}_{\perp}=\mf{U}_{\perp}+\Delta\mf{U}_{\perp}$ and $\mf{\widehat{\Gamma}}_n=\mf{\Gamma}_n+\Delta\mf{\Gamma}_n$. The perturbed matrices satisfy:
$(\mf{T}_{ca}+\Delta\mf{T}_{ca})(\mf{U}_{\perp}+\Delta\mf{U}_{\perp})=(\mf{U}_{\perp}+\Delta\mf{U}_{\perp})(\mf{\Gamma}_n+\Delta\mf{\Gamma}_n).$ 
The perturbation analysis in \cite{wang2017} hinges on (i) the perturbations being \emph{``small enough"} and (ii) ignoring the higher order perturbation terms such as $\Delta\mf{T}_{ca}\Delta\mf{U}_{\perp}$ etc. One of the key drawbacks of this analysis is that a rigorous characterization of an upper bound on the \emph{``small enough perturbation"} $\Vert \Delta\mf{T}_{ca}\Vert_2 \leq \epsilon_1 $ has not been provided explicitly. Secondly, \cite[Theorem~1]{wang2017} makes a critical assumption that ``the signal subspace and the noise subspace are well-separated". This assumption 
{leaves open the possibility of problematic (unidentifiable) source configurations, which have not been explicitly} 
addressed in their analysis. We address both of the aforementioned issues by adopting a non-asymptotic analysis framework that is free from any approximations. Our analysis also explicitly characterizes source configurations that ensure separation between the so-called signal and noise subspaces. In \cite{qiao2019guaranteed}, the first rigorous non-asymptotic probabilistic guarantees were provided for support recovery using a grid-based model. Although their analysis is valid for $S>P$, the sample complexity $L=\Omega(P^2)$ is conservative when $S<P$ as our analysis in \cref{sec:separation} will show.
\section{Performance Analysis of Coarray ESPRIT with Finite Snapshots} \label{sec:performance}

{The Coarray ESPRIT algorithm, an adaptation of ESPRIT in the coarray domain, was introduced in \cite{steinwandt2017performance}. It applies ESPRIT on the estimated coarray covariance matrix  $\mf{\widehat{T}}_{ca}$ as opposed to covariance matrix $\mf{\widehat{R}}_{y}$ of the physical measurements.} For a self-contained exposition, we review the Coarray ESPRIT algorithm and point out certain invariance properties of Coarray ESPRIT. We describe Coarray ESPRIT for the ideal coarray covariance matrix $\mathbf{T}_{ca}$. The extension to the sample covariance estimate is straightforward.
\vspace{-0.3cm}
\subsection{The Coarray ESPRIT Algorithm}\label{sec:co_esp_algo}
The coarray signal subspace is defined as the span of the steering vectors:
$\mathcal{S}_{\text{ca}} := \mathcal{R}\left(\mf{A}_{\mathbb{U}_{\mathbb{S}}}(\blds{\theta})\right). $ Matrix $\mathbf{T}_0:=\mf{A}_{\mathbb{U}_{\mathbb{S}}}(\blds{\theta})\mf{P}\mf{A}_{\mathbb{U}_{\mathbb{S}}}(\blds{\theta})^H$ is positive semi-definite and permits the following eigendecompostion: $\mf{T}_0 = \mathbf{B\Gamma}\mathbf{B}^H$, where the diagonal of $\mf{\Gamma}$ comprises of the eigenvalues ordered in non-increasing fashion and $\mf{B}$ is a unitary matrix. We can partition $\mathbf{B}$ as $\mathbf{B}= [\mathbf{U},\mathbf{U}_{\perp}]$, where the columns of $\mathbf{U} \in \mathbb{C}^{(M_{ca}+1) \times S}$ denote the eigenvectors of $\mf{T}_{0}$ corresponding to its non-zero eigenvalues. Following this decomposition, we write $\mathbf{T}_{ca}$ as:
\begin{equation}\label{Teigdecomp}
\mathbf{T}_{ca} = \mf{T}_0+\sigma^2\mf{I}_{M_{ca}+1}=\mathbf{B}(\mathbf{\Gamma}+ \sigma^2\mathbf{I}_{M_{ca}+1})\mathbf{B}^H.
\end{equation}
If $M_{ca}\geq S$, the Vandermonde structure of $\mf{A}_{\mathbb{U}_{\mathbb{S}}}(\blds{\theta})$ allows us to argue that $\text{rank}( \mf{A}_{\mathbb{U}_{\mathbb{S}}}(\blds{\theta})\mf{P}\mf{A}_{\mathbb{U}_{\mathbb{S}}}(\blds{\theta})^H)=S$, and hence:
\begin{equation}\label{rangespequal}
 \mathcal{S}_{\text{ca}}=\mathcal{R}(\mf{A}_{\mathbb{U}_{\mathbb{S}}}(\blds{\theta})) = \mathcal{R}(\mf{A}_{\mathbb{U}_{\mathbb{S}}}(\blds{\theta})\mf{P}\mf{A}_{\mathbb{U}_{\mathbb{S}}}(\blds{\theta})^H) = \mathcal{R}(\mathbf{U}).   
\end{equation}
As a result of \eqref{rangespequal}, $\exists$ an invertible $\mathbf{Q} \in \mathbb{C}^{S \times S}$ such that
\begin{equation} \label{writingintermsof}
    \mathbf{U} = \mf{A}_{\mathbb{U}_{\mathbb{S}}}(\blds{\theta})\mathbf{Q}.
\end{equation}
 Let $\mathbf{U}_0\in \mathbb{C}^{M_{\text{ca}}\times S}$ and $\mathbf{U}_1 \in \mathbb{C}^{M_{\text{ca}}\times S}$ denote the submatrices corresponding to the first and last $M_{ca}$ rows of $\mathbf{U}$. Similarly, let $\mathbf{V}_0,\mf{V}_1 \in \mathbb{C}^{M_{\text{ca}}\times S}$ be the submatrices corresponding to the first and last $M_{ca}$ rows of $\mf{A}_{\mathbb{U}_{\mathbb{S}}}(\blds{\theta})$. Due to the Vandermonde structure of $\mf{A}_{\mathbb{U}_{\mathbb{S}}}(\blds{\theta})$, the following holds:
$\mf{V}_1=\mathbf{V}_0\mathbf{D}$,
 where 
$
 \mathbf{D} = \text{diag}(e^{j\pi \sin(\theta_1)},e^{j\pi \sin(\theta_2)},\hdots,e^{j\pi \sin(\theta_S)} ).$
 By \eqref{writingintermsof}, matrices $\mf{U}_0$ and $\mf{U}_1$ satisfy:
 \begin{align}
 \mathbf{U}_0 = \mathbf{V}_0\mathbf{Q}, \quad \mathbf{U}_1 =\mathbf{V}_0\mathbf{D}\mathbf{Q}. \label{eqn:U0_main_def}
 \end{align}
Now, consider the matrix 
\ben
\mathbf{\Psi} = \mathbf{U}_0^{\dagger}\mathbf{U}_1 \in \mathbb{C}^{S \times S}.\label{eqn:Psi_def}
\een
Since $\mf{U}_0$ has full column rank \eqref{eqn:Psi_def} implies $\mathbf{U}_0^{\dagger}=\mf{Q}^{-1}\mf{V}_0^{\dagger}$. Plugging this in \eqref{eqn:Psi_def} and combining with \eqref{eqn:U0_main_def}, we have: $\mathbf{\Psi} = \mathbf{Q}^{-1}\mathbf{D}\mathbf{Q}.$
 Hence, the DOAs can be inferred from the eigenvalues of $\mathbf{\Psi}$. Since $L$ is finite, we do not have access to $\mathbf{T}_{ca}$ and Coarray ESPRIT is instead applied on its estimate $\mathbf{\widehat{T}}_{ca}$ defined in \eqref{eqn:T_hat_est}. If we can ensure that the error $\mathbf{E}_L$ is small enough (which we will rigorously specify using Weyl's inequality), $\mathbf{\widehat{T}}_{ca}$ will be at least rank-$S$. Let $\hat{\mathbf{U}}$ be the matrix of eigenvectors corresponding to the largest $S$ eigenvalues of $\mathbf{\widehat{T}}_{ca}$ (which is well-defined). We can consider $\mathbf{\widehat{U}}$ as a basis of the perturbed coarray signal space $\mathcal{\hat{S}}_{\text{ca}}$. From $\mathbf{\widehat{U}}$, we compute the matrices $\mathbf{\widehat{U}}_0$, $\mathbf{\widehat{U}}_1$, $\mathbf{\widehat{\Psi}}$ following the same construction as $\mf{U}_0,\mf{U}_1$ and $\mf{\Psi}$. Let $\hat{\lambda}_i=r_ie^{j\widehat{\phi}_i}$ be the polar representation of the eigenvalues of the matrix $\hat{\mathbf{\Psi}}$. The estimated normalized frequencies $\hat{\Omega}=\{\widehat{\omega}_i\}_{i=1}^{S}$ are then given by
$\widehat{\omega}_i=\frac{\widehat{\phi}_i}{2\pi}.$
\subsection{Basis Invariance Property of ESPRIT}
In the previous section, ESPRIT is performed using the basis given by the singular vectors $\widehat{\mf{U}}$ ($\mf{U}$) of $\widehat{\mf{T}}_{\text{ca}}$ ($\mf{T}_{\text{ca}}$). However, the following Lemma shows that the output of ESPRIT is invariant to the choice of the basis for the subspace.
\begin{lem}\label{lem:ESP_inv}
 Let $\widetilde{\mf{U}}\in \mathbb{C}^{(M_{ca}+1)\times S}$ be another basis for $\mathcal{R}(\widehat{\mf{U}})$. Then, the matrix $\widetilde{\mf{\Psi}}:=\widetilde{\mf{U}}_0^{\dagger}\widetilde{\mf{U}}_1$ is similar to the matrix $\mathbf{\widehat{\Psi}}$, i.e., $\widetilde{\mf{\Psi}}$ and $\mathbf{\widehat{\Psi}}$ share the same eigenvalues. 
\end{lem}
\vspace{-0.2cm}
\begin{proof}
Since $\mathcal{R}(\widetilde{\mf{U}}) = \mathcal{R}(\widehat{\mf{U}})$, there exists an invertible matrix $\mf{W}\in \mathbb{C}^{S\times S}$ such that
$\widetilde{\mf{U}}:=\widehat{\mf{U}}\mf{W}$.
Thus, the following holds:
$\widetilde{\mf{U}}_0=\widehat{\mf{U}}_0\mf{W}, \widetilde{\mf{U}}_1=\widehat{\mf{U}}_1\mf{W}$.
Since $\mf{W}$ is an invertible matrix, $\widetilde{\mf{U}}_0^{\dagger}=\mf{W}^{-1}\widehat{\mf{U}}_0^{\dagger}$ 
and $\widetilde{\mf{\Psi}}=\widetilde{\mf{U}}_0^{\dagger}\widetilde{\mf{U}}_1= \mf{W}^{-1}\widehat{\mf{U}}_0^{\dagger}\widehat{\mf{U}}_1\mf{W}=\mf{W}^{-1}\mathbf{\widehat{\Psi}}\mf{W}$. 
This completes the proof.
\end{proof}
\subsection{Covariance Estimation Error}\label{sec:Toep}
In this section, we obtain tail bounds on $\| \mf{E}_{L} \|_2$ in terms of array parameters in a finite snapshot setting. Such a bound brings out the effect of the array geometry on the estimation error. {Our analysis leverages recent results derived in \cite{Musco} which we specialize for complex Toeplitz Hermitian matrices. Some of our intermediate steps depart from \cite{Musco} by invoking a result on the bounding the supremum of a certain spectral function from \cite{zygmund2002trigonometric}.}
We first introduce the key quantities and intermediate results on bounding the spectral norm of a Toeplitz Hermitian matrix from \cite{zygmund2002trigonometric,gray2006toeplitz}. 

Let $\mf{M}\in \mathbb{C}^{N \times N}$ be any Hermitian symmetric Toeplitz matrix. Such a matrix can be completely described by only its first column. Consider the ``spectral function" associated with $\mf{m}=[m_{-(N-1)},\dots,m_{-1},m_0,m_1,\dots,m_{N-1}]^{\top}$ \cite{gray2006toeplitz}:
${f_{\mf{m}}(\theta)=\sum_{k=-(N-1)}^{N-1} m_k \exp(-jk\theta)}$, 
where $m_k=M_{k+1,1}$ and $m_{-k}=m_{k}^{*}$ as a result of the Hermitian Toeplitz structure. Evidently, the spectral function is a trigonometric polynomial {of order $N-1$ \cite{zygmund2002trigonometric}} whose coefficients are determined by the vector $\mf{m}$. This spectral function $f_{\mathbf{m}}(\theta)$ can be used to bound $\Vert\mathbf{M}\Vert_2$ as indicated by the following lemma from \cite{gray2006toeplitz,zygmund2002trigonometric,Musco}:
\begin{lem} \label{lem:spectral_norm}Let $\mathbf{M} \!\in\! \mathbb{C}^{N \times N}$ be a Hermitian symmetric Toeplitz matrix and $f_{\mathbf{m}}$ be the associated spectral function. Then, 
$\| \mathbf{M} \|_2 \le \sup_{\theta \in [-\pi, \pi]} |f_{\mathbf{m}}(\theta)|$.
\end{lem}
\cref{lem:spectral_norm} indicates that the spectral norm of a Hermitian symmetric Toeplitz matrix can be bounded by the supremum of its associated spectral function. Note that the covariance estimation error $\mathbf{E}_L=\mf{T}_{ca}-\mf{\widehat{T}}_{ca}$ is a Toeplitz Hermitian matrix, satisfying $\mathbf{E}_L=\mathcal{T}(\mf{e})$, where $\mf{e}=[e_{-M_{ca}},\dots,e_{-1},e_0,e_1,\dots,e_{M_{ca}}]^{T}$ is conjugate symmetric and $e_i=t_i-\hat{t}_i$. Therefore, to bound $\| \mf{E}_{L} \|_2$ using \cref{lem:spectral_norm}, we need to investigate the spectral function $f_{\mf{e}}(\theta)$:
\ben\label{eqn:f_e}
f_{\mf{e}}(\theta) := \sum_{k=-M_{ca}}^{M_{ca}} e_k \exp(-j\theta k) .  
\een
Towards this purpose, define $\Lamth$, for $1\leq m,n \leq P$,
\begin{align}\label{eqn:Lambda_def}
    [\Lamth]_{m,n}=\frac{1}{\vert\Omega_{d_{m}-d_{n}}\vert}\exp(j (d_m-d_n)\theta), 
\end{align}
and $\Ey:=\Ry-\Ryov$, where $\Ry$ and $\Ryov$ are defined in \eqref{eqn:Ry_element} and \eqref{eqn:Ryov_element}, respectively. The elements of $\Ey$ are given by:
\ben
[\Ey]_{m,n}=t_{d_m-d_n}-\widehat{t}_{d_{m}-d_{n}}=e_{d_m-d_n}, \ 1\leq m,n \leq P. \label{eqn:Ey_def}
\een
\cref{prop:spec_func} provides a compact representation of $f_{\mf{e}}(\theta)$.
\begin{prop}\label{prop:spec_func}
Let $f_{\mf{e}}(\theta)$ be the spectral function defined in \eqref{eqn:f_e}. Then, the following equality holds: $f_{\mf{e}}(\theta)=\text{tr}\left(\Ey\Lamth\right)$ where $\Lamth$ and $\Ey$ are defined in \eqref{eqn:Lambda_def} and \eqref{eqn:Ey_def}, respectively. 
\end{prop}
\vspace{-0.2cm}
\begin{proof}
{\footnotesize
\begin{align*}
&\text{tr}\left(\Ey\Lamth\right)=\sum_{m,n=1}^{P}[\Ey]_{m,n}[\Lamth]_{n,m}=\sum_{m,n=1}^{P}e_{d_m-d_n}\frac{e^{-j (d_m-d_n)\theta}}{\vert \Omega_{d_m-d_n}\vert}\\
   &=\sum_{s=-M_{ca}}^{M_{ca}}\sum_{\scriptstyle m,n\atop d_m-d_n=s} e_{s}\frac{\exp(-js\theta)}{\vert \Omega_s\vert}=\sum_{s=-M_{ca}}^{M_{ca}}e_{s}\exp(-js\theta).
\end{align*}}\vspace{-0.5cm}
\end{proof}
We introduce a quantity referred to as ``\emph{Redundancy coefficient}" that will play an important role in bounding $\| \mf{E}_{L} \|_2$.
\begin{definition}[Redundancy Coefficient]
Given a hole-free sparse array $\mathbb{S}$, let $M_{ca}$ be the largest element in its difference set $\mathbb{D}_\mathbb{S}$. The redundancy coefficient $\Delta (\mathbb{S})$ is defined as:
$\smash{\Delta (\mathbb{S}):=\sum_{i=0}^{M_{ca}} \frac{1}{\vert \Omega_{i}\vert}}$, 
where set $\Omega_i$ is defined in \cref{def:weight_func}.
\end{definition}
The quantity $\Delta (\mathbb{S})$ is controlled by the redundancy pattern of the sparse array $\mathbb{S}$, i.e., the number of times an element repeats in the difference set.
We provide an illustrative example to show how the quantity $\Delta (\mathbb{S})$ grows as a function of $P$. 




\begin{lem}\label{lem:DS_nested}
Given a generalized nested array $\mathbb{S}^{(N_1,N_2)}_{\text{nest}}$ with $P\!:=\!N_1+N_2\geq 3$ sensors, the following holds:
   $\ln(P) \leq \Delta (\mathbb{S}^{(N_1,N_2)}_{\text{nest}})\leq 2\ln(P)$, if $N_2=1$, and $P^2/16 \leq \Delta (\mathbb{S}^{(N_1,N_2)}_{\text{nest}})\leq P^2$, if $N_1=\lceil P/2 \rceil$ and $N_2=\lfloor P/2 \rfloor \geq 2$.
\end{lem}
\begin{proof}
\textbf{Case I} ($N_2=1$):
The choice $N_2=1$ corresponds to a ULA, with $P = N_1+1$ sensors and $\vert \Omega_i\vert =P-i, i\geq 0$. Therefore,
$
 \Delta (\mathbb{S}^{(P-1,1)}_{\text{nest}})=\sum_{i=0}^{P-1}  \frac{1}{P-i}.$ Such a harmonic sum can be bounded as $\ln(P)\leq \sum_{i=0}^{P-1}  \frac{1}{P-i} \leq 1+\ln(P)$ \cite{Harmonic}. For $P\geq 3$, we get the desired bound since $1+\ln(P)\leq 2\ln(P)$. 

\textbf{Case II} ($N_2=\lfloor P/2 \rfloor \geq 2$):
The differences between the elements of the outer and inner ULA which are of the form $k=i(\lceil P/2\rceil+1)-j$, $2\leq i \leq \lfloor P/2\rfloor$ and $1\leq j \leq \lceil P/2\rceil$, satisfy $\vert \Omega_k\vert=1$. Therefore, we have
$\Delta (\mathbb{S}^{(N_1,N_2)}_{\text{nest}}) \geq \lceil P/2 \rceil \lfloor P/2 \rfloor/2 \geq (P^2/8-P/8)\geq P^2/16$, where the first inequality follows from $ \lfloor P/2 \rfloor-1\geq \lfloor P/2 \rfloor/2$ and the last inequality uses $P\leq P^2/2$ for $P\geq 2$. Since $\mathbb{S}^{(N_1,N_2)}_{\text{nest}}$ is hole free, it implies $1/\vert \Omega_i\vert \leq 1$ for all $0\leq i \leq M_{\text{ca}}$. Therefore, we can bound $\Delta (\mathbb{S}^{(N_1,N_2)}_{\text{nest}})\leq M_{\text{ca}}+1\leq P^2$. 
\end{proof}
As the following Theorem will show, $\Delta (\mathbb{S})$ determines the sample complexity for controlling the covariance estimation error. Therefore, with the same number of sensors, two different array geometries could require drastically different sample complexity for ensuring that the covariance estimation error is bounded by the same quantity with high probability. 
\begin{thm}\label{thm:coverr}
Consider the measurement model \eqref{measurementmodel} obeying assumptions [\textbf{A1-A2}], where $\mathbb{S}$ is a hole-free sparse array with redundancy coefficient $\Delta(\mathbb{S})$. Let $\mf{T}_{ca}\in \mathbb{C}^{M_{ca}+1 \times M_{ca}+1}$ be the coarray covariance matrix defined in \eqref{eqn:ToeplitzTca} and $\mf{\widehat{T}}_{ca}$ be its estimate given by \eqref{eqn:T_hat_est}. For any $\epsilon \ge 0$, we have 
{\footnotesize\begin{align*}
&P\left(\| \mathbf{T}_{ca} - \mathbf{\widehat{T}}_{ca} \|_2 \geq \epsilon \right) \nonumber \\
&\le 8 M_{ca} \exp \left[ -c_1 L \min \left( \frac{c_2\epsilon^2 }{\| \mathbf{R_{y}} \|_2^2 \Delta(\mathbb{S})}, \frac{\epsilon }{\| \mathbf{R_{y}} \|_2 \sqrt{\Delta(\mathbb{S})}}\right) \right],
\end{align*}}
where $c_1$ and $c_2$ are a positive universal constants.
\end{thm}
\vspace{-0.2cm}
\begin{proof}
The proof is in Appendix \ref{app:thm1}.
\end{proof}
\vspace{-0.5cm}
\subsection{Frequency/Angle Estimation Error of Coarray ESPRIT}\label{sec:frq}


We next bound the DOA estimation error in terms of the covariance estimation error $\EL$. Finally, we will combine this bound with the probabilistic bounds on $\Vert\EL\Vert_2$ in \cref{thm:coverr} to obtain the main sample complexity result (in \cref{thm:main_coesp}). We will use the matching distance metric, defined as follows \cite{WeilinEsprit}:
\begin{equation}
    \text{md}(\blds{\theta},\hat{\blds{\theta}}) := \min_{\Pi \in \mathcal{P}}\ \max_j \ \min_{k\in \mathbb{Z}}\vert \hat{\omega}_{\Pi(j)} - \omega_j+k  \vert
\end{equation}
{where $\omega_i$ ($\hat{\omega}_i$) are the normalized DOAs and $\mathcal{P}$ denotes the set of all possible permutations on $\{1,2,\cdots,S\}$}.


For our analysis, we will use an additional assumption that will be invoked whenever suitable:
\begin{itemize}
   \item[] \textbf{[A3]} The number of sources $S=O(1)$, i.e., $S$ is held constant and does not grow with $P$.
\end{itemize}
\textbf{Eigen Gap condition:}
Define:
\begin{align} \label{eqn:beta}
\beta:=p_{\min}\sigma_S^2(\mf{A}_{\mathbb{U}_{\mathbb{S}}}(\blds{\theta}))-\sigma^2.
\end{align}
Henceforth, we will refer the condition $\beta>0$ as the \emph{``eigen gap condition"} and it will play an important role in our analysis. Recall, from the definition of $\mf{T}_{ca}=\mf{A}_{\mathbb{U}_{\mathbb{S}}}(\blds{\theta})\mf{P}\mf{A}_{\mathbb{U}_{\mathbb{S}}}(\blds{\theta})^H+\sigma^2$, $\beta > 0$ ensures that there is a margin between the smallest singular value of $\mf{A}_{\mathbb{U}_{\mathbb{S}}}(\blds{\theta})\mf{P}\mf{A}_{\mathbb{U}_{\mathbb{S}}}(\blds{\theta})^H$ and the $(S+1)^\text{th}$ singular value of $\mathbf{T}_{ca}$ (determined by the noise $\sigma$) as $p_{\min}\sigma_{S}^2(\mf{A}_{\mathbb{U}_{\mathbb{S}}}(\blds{\theta}))$ is a lower bound on $\sigma_S(\mf{T}_{ca})$. The following theorem relates the DOA estimation error in terms of matching distance to the covariance estimation error $\mathbf{E}_L$, provided the latter is upper bounded by a suitable quantity. 

\begin{thm}\label{Espiriterror}
Let $\mathbb{S}$ be a hole-free sparse linear array with $P$ sensors. Let $\mf{T}_{ca}\in \mathbb{C}^{M_{ca}+1 \times M_{ca}+1}$ be the coarray covariance matrix defined in \eqref{eqn:ToeplitzTca} and $\mf{\widehat{T}}_{ca}$ be its estimate given by \eqref{eqn:T_hat_est}. If assumption [\textbf{A3}] holds and the following conditions are satisfied: 
\begin{equation}\label{eqn:SNRcondition1}
  \beta>0 \quad \text{ and } \Vert \mathbf{E}_L \Vert_2 \leq C_S \beta
\end{equation}
then the matching distance error of ESPRIT algorithm satisfies
\begin{equation}\label{matcherror2}
\text{md}(\blds{\theta},\hat{\blds{\theta}}) \leq q \Vert \mathbf{E}_L \Vert_2
\end{equation}
where $\mf{E}_L$, $\beta$ are defined in \eqref{eqn:EL}, \eqref{eqn:beta}, $q= (C_S' \sqrt{M_{ca}+1})/(\beta\sigma_S(\mf{A}_{\mathbb{U}_{\mathbb{S}}}(\blds{\theta})))$. Quantities $C_S,C_S'$ ~are dependent only on $S$ which is assumed to be $O(1)$.
\end{thm}
\vspace{-0.2cm}
\begin{proof}
See Appendix \ref{app:ESP}.
\end{proof}

The following Lemma obtains both lower and upper bounds on the spectral norm $\Vert \mf{R}_{\mf{y}}\Vert_2$ that are valid regardless of the array geometry. 
\begin{lem}\label{lem:Ry_up_bound}
Consider the covariance matrix $ \mf{R}_{\mf{y}}$ given by \eqref{eqn:Ry_gt}, where $\mathbb{S}$ is any (sparse) array. Given a fixed $S$, signal powers $\mf{p}$ and noise power $\sigma^2$, for all $\blds{\theta}$ the following holds:
\begin{equation}
\smash{p_{\min}P\leq \smash{\Vert \mf{R}_{\mf{y}}\Vert_2 \leq p_{\max}PS +\sigma^2.}} \label{eqn:Ry_bnd}
\end{equation}
\end{lem}
\begin{proof}
For any $\mathbb{S}$, we can bound the spectral norm $\Vert \mf{R}_{\mf{y}}\Vert_2$ as:
\begin{align*}
    \Vert \mf{R}_{\mf{y}}\Vert_2 &=\sigma_1(\mathbf{A}_{\mathbb{S}}(\blds{\theta})\mf{P}\mathbf{A}_{\mathbb{S}}(\blds{\theta})^H)+\sigma^2\leq p_{\max}\sigma_1(\mathbf{A}_{\mathbb{S}}(\blds{\theta}))^2+\sigma^2\\
    &\leq p_{\max}PS +\sigma^2 
\end{align*}
where the last inequality follows from the fact that $\sigma_1(\mathbf{A}_{\mathbb{S}}(\blds{\theta}))^2\leq \Vert\mathbf{A}_{\mathbb{S}}(\blds{\theta})\Vert_F^2 = PS $. Similarly, we can lower bound the norm
$\Vert \mf{R}_{\mf{y}}\Vert_2 \geq \sigma_1(\mathbf{A}_{\mathbb{S}}(\blds{\theta})\mf{P}\mathbf{A}_{\mathbb{S}}(\blds{\theta})^H)\geq p_{\min}\sigma_1^2(\mf{A}_{\mathbb{S}}(\blds{\theta}))\geq p_{\min}\Vert \mf{A}_{\mathbb{S}}(\blds{\theta})\Vert_F^2/S = p_{\min}P$.
\end{proof}
Combining \cref{thm:coverr} and \ref{Espiriterror}, we next present a sufficient condition on the number ($L$) of snapshots in terms of the model parameters (array geometry, SNR and source configuration) that allows us to bound the matching distance error by a prescribed $\epsilon$ with probability at least $1-\delta$.
\begin{thm}\label{thm:main_coesp}
Consider the measurement model \eqref{measurementmodel}, where $\mathbb{S}$ is a hole-free sparse array. Suppose $\beta>0$ and the statistical assumptions [\textbf{A1-A3}] hold. Then for any $0<\delta <1$ and $\epsilon >0$, the matching distance error satisfies $\text{md}(\blds{\theta},\hat{\blds{\theta}}) \leq \min(\epsilon,C_S\beta q)$  with probability at least $1-\delta$, provided 
{  \begin{equation}
   L\!\geq\!c_3\ln\bigg(\frac{8M_{ca}}{\delta}\bigg) \max\!\left(\frac{q^2_1 \Delta(\mathbb{S})}{c_2\epsilon^2},\!\frac{q_1 \sqrt{\Delta(\mathbb{S})}}{\epsilon},\!\frac{L_0^2}{c_2},\!L_0\!\right)\!.\label{eqn:l_bound}
\end{equation}}
Here $q_1=q\|\mf{R}_{\mf{y}}\|_2,L_0= \| \mathbf{R_{y}} \|_2 \sqrt{\Delta(\mathbb{S})}/(C_S\beta)
$ and $c_2,c_3$ are universal constants.
\end{thm}
\vspace{-0.2cm}
\begin{proof}
See Appendix \ref{app:prob_coesp}
\end{proof}

\begin{cor}\label{cor:main_coesp}
Consider the measurement model \eqref{measurementmodel}, where $\mathbb{S}$ is a hole-free sparse array. Suppose $\beta>0$ and the statistical assumptions [\textbf{A1-A3}] hold. Then for any $0<\delta < 1$ and $0<\epsilon  \leq q\min( C_S \beta, p_{\min}P\sqrt{\Delta(\mathbb{S})}/c_2 )$, the matching distance error satisfies $\text{md}(\blds{\theta},\hat{\blds{\theta}}) \leq \epsilon$  with probability at least $1-\delta$ provided
  \begin{align}
   &L \geq c_3\ln\left({8M_{ca}}/{\delta}\right) {q^2_1 \Delta(\mathbb{S})}/{(c_2\epsilon^2)},\label{eqn:l_bound2}
\end{align} 
where $q_1,L_0
$,$c_2,c_3$ are given in \cref{thm:main_coesp}.
\end{cor}
\vspace{-0.2cm}
\begin{proof}
Using the lower bound on $\Vert \mf{R_y}\Vert_2$ from \cref{lem:Ry_up_bound}, we can see $\epsilon \leq \min(C_S\beta q, q_1\sqrt{\Delta(\mathbb{S})}/c_2)$. Since $\beta\geq \epsilon/(C_Sq)$, this implies 
$ L_0 \leq {q_1 \sqrt{\Delta(\mathbb{S})}}/{\epsilon}$. This inequality also implies $L_0^2/c_2 \leq  {q^2_1 \Delta(\mathbb{S})}/{(c_2\epsilon^2)}$. Using $\epsilon \leq q_1\sqrt{\Delta(\mathbb{S})}/c_2$, we can conclude that $L_0\leq ({q_1 \sqrt{\Delta(\mathbb{S})}}/{\epsilon^2})(q_1\sqrt{\Delta(\mathbb{S})}/c_2)=\frac{q_1^2\Delta(\mathbb{S})}{c_2\epsilon^2}$. Therefore, \eqref{eqn:l_bound2} implies \eqref{eqn:l_bound} since $\max ( \frac{q^2_1 \Delta(\mathbb{S})}{c_2\epsilon^2}, \frac{q_1 \sqrt{\Delta(\mathbb{S})}}{\epsilon},\frac{L_0^2}{c_2},L_0)\!=\!\frac{q^2_1 \Delta(\mathbb{S})}{c_2\epsilon^2}$, and the proof is completed.
\end{proof}
\textbf{Role of redundancy coefficient {in determining Temporal Sample Complexity}:}
\cref{cor:main_coesp} indicates that if the number of snapshots grows proportional to the redundancy coefficient $\Delta(\mathbb{S})$, then it is possible to bound the matching distance error by an arbitrarily small $\epsilon$. Recall that $\Delta(\mathbb{S})$ is a function of the redundancy pattern of $\mathbb{S}$ and from \cref{lem:DS_nested} we have $\Delta(\mathbb{S}_{\text{ula}})=\Theta(\ln(P))$ and $\Delta(\mathbb{S}_{\text{nest}})=\Theta(P^2)$. Based on this, at a cursory glance, one may be tempted to conclude from \eqref{eqn:l_bound2} that for the same number of sensors, the snapshot requirement for the nested array is significantly larger than for the ULA. This is also consistent with an existing misconception that co-array based processing requires a large number of snapshots. However, in reality the sample complexity is also controlled by the interaction of $\Delta(\mathbb{S})$ with other geometry dependent terms in \eqref{eqn:l_bound2} such as $q_1=q\Vert\mf{R}_{\mf{y}}\Vert_2$,{ which in turn depend on both the physical array and coarray size}. {In the next section, 
we clarify this misconception regarding the seemingly higher snapshot requirement of nested arrays in the setting $S=O(1)$.}

\textbf{Spatiotemporal trade-offs:} The snapshot requirement in \cref{cor:main_coesp} is inversely proportional to $\beta$ (since $q\propto \frac{1}{\beta}$). If the array geometry and source configuration are kept fixed and we increase the SNR (either by increasing $p_{\min}$ or decreasing noise power $\sigma$), \cref{cor:main_coesp} suggests that it is possible to achieve the same probability of error with fewer snapshots. Our simulations also are consistent with this theoretical prediction. This SNR and geometry dependent snapshot characterization is another novel contribution of our work.  

\section{A Closer look at the Separation Condition for Super-Resolution with Sparse Arrays}\label{sec:separation}

In order to understand the behavior of the smallest non-zero singular value $\sigma_S(\mf{A}_{\mathbb{U}_{\mathbb{S}}}(\blds{\theta}))$, we consider the notion of minimum separation \cite{WeilinEsprit}:
\begin{equation}\label{mindistancedef}
 \Delta_{\min}(\blds{\theta}) = \min_{\substack{i,j\in \Omega \\i\neq j }} \min_{k \in \mathbb{Z}} \Big\vert \omega_i - \omega_j + k \Big\vert 
 \end{equation}
 where $\omega_i$ is the normalized spatial frequency corresponding to direction $\theta_i$.
By definition, for all $\blds{\theta}$ we have $0 \leq \Delta_{\min}(\blds{\theta}) \leq 1/2$. Instead of analyzing an arbitrary source configuration $\blds{\theta}$, one can obtain a more interpretable condition by representing \eqref{eqn:SNRcondition1} as a function of the minimum separation. The source configurations where $\Delta_{\min}(\blds{\theta})$ is larger than some threshold inversely proportional to $M_{ca}+1$ (i.e. $\Delta_{\min}(\blds{\theta})>\frac{\gamma}{M_{\text{ca}}+1}, \gamma>1$) will be referred to as the ``well-separated" regime. We will inspect what this means for specific array geometries such as the ULA and nested array, and obtain tight bounds on $L$. 
\vspace{-0.3cm}
\subsection{The ``Well-Separated" Case}
In this section, we turn our attention to
how the eigen gap condition can be utilized to obtain sufficient conditions on SNR for different array geometries in the ``well-separated" regime. Let $\mf{V} \in \mathbb{C}^{K\times S}$ be a Vandermonde matrix, with
$
    [\mf{V}]_{m,n}=z_n^{m-1}
$
where $\{z_n\}_{n=1}^{S}$ are the so called ``nodes" of the matrix. We begin by summarizing results from \cite{moitra2015super,batenkov2020conditioning,li2021stable,aubel2019vandermonde} which characterize the minimum singular value of a Vandermonde matrix in the well-separated regime. The following Lemma follows from \cite[Eq. (32)]{aubel2019vandermonde}  which is an intermediate result from  \cite[Theorem 1]{aubel2019vandermonde}. 
\begin{lem}\label{lem:vand_well_sep}
Let $\mf{V}(\blds{\alpha}) \in \mathbb{C}^{K\times S}$ be a Vandermonde matrix with $z_n=e^{j2\pi \alpha_n}$ for $1\leq n \leq S$ and $S\leq K$. If $\alpha_i\in[0,1)$ are all distinct and satisfy: 
\ben
\min_{\substack{i,j\in \Omega \\i\neq j }} \min_{k \in \mathbb{Z}} \Big\vert \alpha_i-\alpha_j + k \Big\vert \geq \frac{\gamma}{K}\label{eqn:alpha_sep}
\een
for some constant $\gamma>1$, then the following holds:
\ben\label{lem:well_sep}
\sigma_S(\mf{V}(\blds{\alpha}))^2 \geq {K}/{C'},\text{ where } C' := {\gamma}/{(\gamma-1)}.
\een
\end{lem}
From \cref{lem:vand_well_sep}, for $\mathbb{S}=\mathbb{S}_{\text{ula}}$ if the source configurations $\blds{\theta}$ satisfies $\Delta_{\min}(\blds{\theta}) \geq \frac{\gamma}{P}$ for some $\gamma>1$ and $S\leq P$ then we have the following lower bound:
\begin{align}
\sigma_S(\mf{A}_{\mathbb{U}_{\mathbb{S}}}(\blds{\theta}))^2 \geq P/C' \label{eqn:ula_sing}
\end{align}
In the following Proposition, we apply \cref{lem:well_sep} to characterize lower bounds on $\sigma_S(\mf{A}_{\mathbb{U}_{\mathbb{S}}})$ for the nested array.
\begin{prop}[Well-Separated]\label{prop:sing}
    Let $\mathbb{S}=\mathbb{S}^{(N_1,N_2)}_{\text{nest}}$ be a nested array with $N_1=\lceil P/2\rceil$ and $N_2=\lfloor P/2\rfloor$ with $P\geq 3$. Suppose $\Delta_{\min}(\blds{\theta}) \geq \frac{5 \gamma}{P^2}$ for some $\gamma>1$ and $S \leq P^2/5$. Then, the following lower bound holds:
    \begin{align}
        \sigma_S(\mf{A}_{\mathbb{U}_{\mathbb{S}}}(\blds{\theta}))^2 \geq P^2/C_n', \text{ where } C_n'={5\gamma}/{(\gamma-1)}.
    \end{align}
\end{prop}
\begin{proof}
For the nested array  with $N_1=\lceil P/2\rceil$ and $N_2=\lfloor P/2\rfloor$, from \eqref{eqn:M_ca_bound} we have $M_{ca}+1 \geq \frac{P^2}{5}$. Hence, $\Delta_{\min}(\blds{\theta}) \geq \frac{5\gamma}{P^2}$ implies $\Delta_{\min}(\blds{\theta}) \geq \frac{\gamma}{M_{\text{ca}}+1}$. Therefore, the condition on $\Delta_{\min}(\blds{\theta})$ in \cref{lem:vand_well_sep} holds and we have the desired lower bound:
    $ \sigma_S(\mf{A}_{\mathbb{U}_{\mathbb{S}}}(\blds{\theta}))^2 \geq \frac{M_{\text{ca}}+1}{C'} \geq (\frac{\gamma-1}{\gamma})\frac{P^2}{5}=\frac{P^2}{C_n'}.
    $
\end{proof}
 \cref{prop:sing} shows that for a nested array, the sources are well-separated if $\Delta_{\min} (\blds{\theta}) \geq 5\gamma/P^2$ and in this case, $\sigma_S(\mf{A}_{\mathbb{U}_{\mathbb{S}}}(\blds{\theta}))$ grows as $\Omega(P)$, owing to the the larger difference coarray of a nested array. 

 In order to highlight the dependence of sample complexity only on key model parameters, we define quantities to combine parameters that are held fixed (such as $S,p_{\min},p_{\max},\sigma$): 
 \begin{align}
 C_{\text{ula}}(S,\sigma,p_{\max}) &:=8C_S^{'2}C^{'3}\frac{c_3}{c_2} (S+\frac{\sigma^2}{p_{\max}})^2 \label{eqn:C_ula}\\
 C_{\text{nest}}(S,\sigma,p_{\max}) &:=4C_S^{'2}C_n^{'3}\frac{c_3}{c_2}(S + \frac{\sigma^2}{p_{\max}})^2 \label{eqn:C_nest}
 \end{align}
 where $C',C_n'$ are universal constants and $C_S$ defined in \cref{Espiriterror} is dependent only on $S$. Using \cref{prop:sing}, we now specialize \cref{cor:main_coesp} for the ULA and nested array.

\begin{thm}\label{thm:well_sep_ula}
Let $\mathbb{S}=\mathbb{S}_{\text{ula}}$ be a ULA with $P$ sensors. Suppose the minimum angular separation between the sources, and the SNR satisfy the following conditions {for some $\gamma>1$:} 
\be
\Delta_{\min}(\blds{\theta})\geq \gamma/P, \quad p_{\min}/{\sigma^2} > {2C'}/{P}, \text{ where }C'= \frac{\gamma}{\gamma-1}.
\ee
 Under assumptions [\textbf{A1-A3}], for any $0<\delta<1$ and $0<\epsilon \leq C_1(S):=C_SC_S'$, $\text{md}(\blds{\theta},\widehat{\blds{\theta}}) \leq \epsilon$ is satisfied with probability at least $1-\delta$, {provided $P\geq 3$ and} 
\begin{align}
    L \geq \frac{{C}_{\text{ula}}(S,\sigma,p_{\max})}{\epsilon^2}\left(\frac{p_{\max}}{p_{\min}}\right)^2\left(\ln \left(\frac{8P}{\delta} \right)\right)^{2}. \label{eqn:snap_bound}
\end{align}
\end{thm}
\begin{proof}
From \cref{lem:vand_well_sep}, if $\Delta_{\min}(\blds{\theta}) \geq \gamma/P$, we have $\sigma_S^2(\mf{A}_{\mathbb{U}_{\mathbb{S}}}(\blds{\theta}))\geq \frac{P}{C'}$. Under the assumption on the SNR, we have $p_{\min} \sigma_S^2(\mf{A}_{\mathbb{U}_{\mathbb{S}}}(\blds{\theta}))\geq  p_{\min}\frac{P}{C'}>2\sigma^2$ which ensures $\beta>p_{\min} \sigma^2_{S}(\mf{A}_{\mathbb{U}_{\mathbb{S}}}(\blds{\theta}))/2>0$. Notice that for ULA $M_{ca}+1=P$ and from the fact that $\sigma_S^2(\mf{A}_{\mathbb{U}_{\mathbb{S}}}(\blds{\theta}))\geq \frac{P}{C'}$, we can obtain the following bound:
\begin{align}
    q &= \frac{C_S' \sqrt{P}}{\beta\sigma_S(\mf{A}_{\mathbb{U}_{\mathbb{S}}}(\blds{\theta}))} \leq \frac{2C_S' \sqrt{P}}{p_{\min}\sigma^3_S(\mf{A}_{\mathbb{U}_{\mathbb{S}}}(\blds{\theta}))}\leq \frac{C_S''}{p_{\min}P} 
    \label{eqn:q_bnd1}
\end{align}
where $C_S''=2C_S'C'^{1.5}$. Notice that:
\begin{align}
    C_S\beta q =\frac{C_SC_S' \sqrt{M_{ca}+1}}{\sigma_S(\mf{A}_{\mathbb{U}_{\mathbb{S}}}(\blds{\theta}))} \geq C_1(S)\frac{\sqrt{P}}{\sqrt{P}}= C_1(S)\label{cond1} 
\end{align}
where the inequality follows from $\sigma_S(\mf{A}_{\mathbb{U}_{\mathbb{S}}})\!\leq\! \Vert\mf{A}_{\mathbb{U}_{\mathbb{S}}}\Vert_F/\sqrt{S}\! =\! \sqrt{P}$. {Using the fact that  $\beta \leq p_{\min}\sigma^2_S(\mf{A}_{\mathbb{U}_{\mathbb{S}}}(\blds{\theta}))$, and the above lower bound on $\sigma_S(\mf{A}_{\mathbb{U}_{\mathbb{S}}}(\blds{\theta}))$, we obtain}
\ben\label{eqn:cond2}
\smash{q\geq \frac{C_S' \sqrt{P}}{p_{\min}\sigma^3_S(\mf{A}_{\mathbb{U}_{\mathbb{S}}}(\blds{\theta}))} \geq \frac{C_S' }{p_{\min}P}}.
\een
{Therefore, $q p_{\min}P\sqrt{\Delta(\mathbb{S}_{\text{ula}})}/c_2 \!\geq\! C_S' \sqrt{\Delta(\mathbb{S}_{\text{ula}})}/c_2 \!\geq\! C_S' \sqrt{\ln(P)}/c_2$, where the last inequality follows from the lower bound on $\Delta(\mathbb{S}_{\text{ula}})$ in \cref{lem:DS_nested}. Recall that $c_2 <1$ \footnote{The constant $c_2=3/16\sqrt{2}$ is specified in the proof of Theorem 1 in Appendix A.},
and therefore for $P\geq 3$,  $\sqrt{\ln P}/c_2 >1$. This implies that $\min(C_1(S),C_S' \sqrt{\ln(P)}/c_2)=C_1(S)$.} 
Combining this with \eqref{cond1}, we have $\epsilon \leq C_1(S)= \min(C_1(S),C_S' \sqrt{\ln(P)}/c_2)\leq \min(C_S\beta q,qp_{\min}P\sqrt{\Delta\mathbb{S}_{\text{ula}}}/c_2)$, which ensures that the assumption on $\epsilon$ in \cref{cor:main_coesp} holds. From \cref{lem:Ry_up_bound}, we have $\Vert \mf{R}_{\mf{y}}\Vert_2 \leq p_{\max}PS+\sigma^2$. {Using this bound and 
\eqref{eqn:q_bnd1}, we get:}
\begin{align}
 q_1 \sqrt{\Delta (\Sula)} &\leq \frac{C_S''}{p_{\min}P}(PSp_{\max}+\sigma^2)\sqrt{2\ln(P)}\nonumber\\
 &=C_S'' (S+\frac{\sigma^2}{p_{\max}P})\left(\frac{p_{\max}}{p_{\min}}\right)\sqrt{2\ln(P)}\nonumber\\
 &\leq\widetilde{C}_1(S,\sigma,p_{\max})\left(\frac{p_{\max}}{p_{\min}}\right)\sqrt{\ln(8P/\delta)}\label{eqn:q1_bound2}
\end{align}
where $\widetilde{C}_1(S,\sigma,p_{\max}):=(S+\frac{\sigma^2}{p_{\max}})\sqrt{2}C_S''$. The upper bound follows from the observations that $(S+\frac{\sigma^2}{p_{\max}P})\leq (S+\frac{\sigma^2}{p_{\max}})$ for all $P\geq 1$ and $\ln(P)\leq \ln(8P/\delta)$ for any $\delta<1$. { Notice from  \eqref{eqn:C_ula}, that $C_{\text{ula}}(S,\sigma,p_{\max}) =c_3/c_2 \widetilde{C}^2_1(S,\sigma,p_{\max})$.} From \eqref{eqn:q1_bound2}, we have
\begin{align*}
     &c_3\ln(\frac{8P}{\delta}) \frac{q^2_1 \Delta(\mathbb{S}_{\text{ula}})}{c_2\epsilon^2} \leq \frac{c_3}{c_2\epsilon^2}\widetilde{C}^2_1(S,\sigma,p_{\max})(\frac{p_{\max}}{p_{\min}}\ln(8P/\delta))^2\\
     &=\frac{C_{\text{ula}}(S,\sigma,p_{\max})}{\epsilon^2}\left(\frac{p_{\max}}{p_{\min}}\right)^2\left(\ln(8P/\delta)\right)^2.
\end{align*}
Therefore, \eqref{eqn:snap_bound} implies \eqref{eqn:l_bound2} and the proof is completed by applying \cref{cor:main_coesp} since $\beta>0$ and the conditions on $\epsilon$ and $L$ required for applying the corollary are satisfied.
\end{proof}

\begin{thm}\label{thm:well_sep_nst}
Let $\mathbb{S}=\Snst^{(N_1,N_2)}$ be a nested array with $N_1=\lceil P/2 \rceil$ and $N_2=\lfloor P/2 \rfloor$. Suppose the minimum angular separation between the sources, and the SNR satisfy the following conditions {for some $\gamma>1$:
\be
\Delta_{\min}(\blds{\theta})\geq \frac{5\gamma}{P^2}, \quad  \frac{p_{\min}}{\sigma^2} > \frac{2C_n'}{P^2}, \text{ where } C_n'={5\gamma}/{(\gamma-1)}.
\ee
Under} the assumptions [\textbf{A1-A3}], for any $\delta>0$ and $0<\epsilon\leq C_2(S):=\sqrt{1/5} C_SC_S'$, $\text{md}(\blds{\theta},\widehat{\blds{\theta}}) \leq \epsilon$
is satisfied with probability at least $1-\delta$ { provided $P\geq 3$ and
\begin{align}
    L \geq\frac{C_{\text{nest}}(S,\sigma,p_{\max})}{\epsilon^2}\left(\frac{p_{\max}}{p_{\min}}\right)^2\ln \left(\frac{8P^2}{\delta} \right). \label{eqn:snap_bound2}
\end{align}}
\end{thm}
\begin{proof}
From \cref{prop:sing}, {if $\Delta_{\min}(\blds{\theta}) \geq 5\gamma/P^2$,} we have $\sigma_S^2(\mf{A}_{\mathbb{U}_{\mathbb{S}}}(\blds{\theta}))\geq \frac{P^2}{C_n'}$. { Following the same argument as \cref{thm:well_sep_ula}, this ensures that $\beta >0$. Using the fact that $M_{\text{ca}}+1\leq P^2$ (from \eqref{eqn:M_ca_bound}) and 
the lower bound on $\sigma_S^2(\mf{A}_{\mathbb{U}_{\mathbb{S}}}(\blds{\theta}))$, we obtain 
\begin{equation}
\smash{q \leq\frac{C_S'P}{\beta\sigma_S(\mf{A}_{\mathbb{U}_{\mathbb{S}}}(\blds{\theta}))} \leq \frac{2C_S' P}{p_{\min}\sigma^3_S(\mf{A}_{\mathbb{U}_{\mathbb{S}}}(\blds{\theta}))}\leq \frac{\bar{C}_S''}{p_{\min}P^2}} \label{eqn:q_bnd2}
\end{equation}
where $\bar{C}_S'':=2C_S'C_n'^{1.5}$. Notice that $\sigma_S(\mf{A}_{\mathbb{U}_{\mathbb{S}}}(\blds{\theta})) \leq \|\mf{A}_{\mathbb{U}_{\mathbb{S}}}\|_{F}/\sqrt{S} = \sqrt{M_{\text{ca}}+1} \leq P$. Hence,  similar to \eqref{cond1}, we can establish that $C_S\beta q \geq C_2(S).$  Using the fact $P^2/5\leq M_{\text{ca}}+1$ from \eqref{eqn:M_ca_bound}, similar to \eqref{eqn:cond2} we obtain $q\geq \frac{C_S' P}{\sqrt{5}p_{\min}\sigma^3_S(\mf{A}_{\mathbb{U}_{\mathbb{S}}}(\blds{\theta}))} \geq \frac{C_S' }{\sqrt{5}p_{\min}P^2}$. From \cref{lem:DS_nested}, $\Delta(\mathbb{S}_{\text{nest}}) \geq P^2/16$. It follows that 
$q p_{\min}P\sqrt{\Delta(\mathbb{S}_{\text{nest}})}/c_2 \geq \frac{C_S'}{4c_2\sqrt{5}}$. Since $4c_2<1$, it follows that $\min(C_2(S),  C_S'/(4c_2\sqrt{5}))=C_2(S)$
and therefore $\epsilon \leq C_2(S)=\min( C_2(S),  C_S'/(4c_2\sqrt{5}))$ 
ensures that the assumption on $\epsilon$ in \cref{cor:main_coesp} holds.} 
Using $\Delta (\Snst)\leq P^2$ (from \cref{lem:DS_nested}), \cref{lem:Ry_up_bound},  and \eqref{eqn:q_bnd2}, we get:  
\begin{align}
    q_1 \sqrt{\Delta (\Snst)} \leq \widehat{C}_1(S,\sigma,p_{\max})({p_{\max}}/{p_{\min}}),\label{q_1_nested}
\end{align}
where $\widehat{C}_1(S,\sigma,p_{\max})\!=\!(S + \frac{\sigma^2}{p_{\max}})\bar{C}_S''$. By \eqref{q_1_nested}, we have  
\begin{align*}
     &\ln(\frac{8M_{ca}}{\delta}) \frac{c_3q^2_1 \Delta(\mathbb{S}_{\text{nst}})}{c_2\epsilon^2} \leq \frac{c_3}{c_2\epsilon^2}\widehat{C}^2_1(S,\sigma,p_{\max})\ln(8P^2/\delta)(\frac{p_{\max}}{p_{\min}})^2\\
     &=\frac{C_{\text{nest}}(S,\sigma,p_{\max})}{\epsilon^2}\left(\frac{p_{\max}}{p_{\min}}\right)^2\left(\ln(8P^2/\delta)\right).
\end{align*}
Therefore \eqref{eqn:snap_bound2} implies \eqref{eqn:l_bound2} and the proof is again completed by applying \cref{cor:main_coesp} since $\beta >0$ and the  conditions on $\epsilon$ and $L$ required for applying the corollary are satisfied. 
\end{proof}
Note that the range of values for $\epsilon$ where \cref{thm:well_sep_ula} and \ref{thm:well_sep_nst} are applicable differ slightly. However in the regime $\epsilon \leq \min(C_1(S), C_2(S))=C_2(S)$ and $P\geq 3$, we can fairly compare the two array geometries.

\textbf{Towards higher resolution with same snapshots:} Theorem \ref{thm:well_sep_ula} states that for a ULA, the matching distance error for Coarray ESPRIT can be bounded by $\epsilon$ provided (i) the snapshots scales only (poly)logarithmically in the dimension of the coarray covariance matrix and (ii) the minimum separation is $\Delta_{\min}\geq \gamma/P$. On the other hand, Theorem \ref{thm:well_sep_nst} guarantees that for a nested array with $P$ sensors, it is possible to bound the matching distance error by the same $\epsilon$ with order wise the same number of snapshots ($L=\Omega(\ln(P^2)$), but with a \emph{relaxed separation} condition that allows $\Delta_{\min}$ to be $\Delta_{\min}=\Omega(1/P^2)$. This validates the superior resolution properties of nested arrays compared to ULA with \emph{the same budget of temporal snapshots}. This has been empirically observed in the literature, but never theoretically established, until now. 

\textbf{Noise Resilience of Nested Arrays:}
If we consider the separation regime $\Delta_{\min}=\Omega(1/P)$ that is applicable for both the ULA and nested array, Theorems \ref{thm:well_sep_ula} and \ref{thm:well_sep_nst} indicate that the SNR $(p_{\min}/\sigma^2)$ requirement for the nested array can be $P$ times smaller than that of the ULA, in order to achieve the same DOA error bound with order-wise the same number of snapshots ($L=\Omega(\ln P)$). This brings out another advantage of nested arrays in terms of robustness against noise, especially in the low-SNR regime \cite{shahsavari2021fundamental}. 

\textbf{Effect of Dynamic Range:}
Our analysis also reveals the challenge posed by sources with higher dynamic range $p_{\max}/p_{\min}$ {as also observed in \cite{liao2016music}.} Theorem \ref{thm:well_sep_nst} suggests that at the same SNR (defined with respect to the weakest source $p_{\min}$), more snapshots maybe needed for resolving sources with disproportionately varying powers (higher $p_{\max}$ compared to the fixed $p_{\min}$). As will be shown, the numerical results are indeed consistent with the prediction made by our analysis. 
\vspace{-0.3cm}
\subsection{The Myth of Large Snapshots: Correlation Error vs. Angle Estimation Error}\label{sec:myth}
Since nested (and other) sparse arrays realize the virtual difference coarray by correlation-processing, it is commonly believed that one needs a large number ($L = \Omega (P^2)$) of temporal snapshots to estimate $\Theta(P^2)$ (cross) correlation values between sensor pairs. This `myth' of large snapshots (that grows quadratically in the number of sensors $P$) is partially true, if our goal is to estimate the coarray covariance matrix $\mf{T}_{ca}$. If we only allow $L$ to scale as $L = \Theta (\log P)$ (the so-called sample-starved regime), then one may indeed incur large error in covariance estimation. However, Theorem \ref{thm:well_sep_nst} shows that the angle estimation error can be made arbitrarily small ($\epsilon$) with high probability ($1-\delta$) provided $L$ scales only as $\Omega(\frac{1}{\epsilon^2}\ln(8P^2/\delta))$, despite the possibility of the coarray covariance error of a nested array {increasing with $P$} in this snapshot-starved regime. This surprising phenomenon is due to the fact that the potentially large covariance estimation error (which can even grow with $P$ in this regime) can actually be mitigated/counterbalanced by the enhanced aperture/difference set of the nested array that results in a large restricted smallest singular value $\sigma_S(\mathbf{A}_{\mathbb{U}_{\mathbb{S}}})$. As long as $\Delta_{\min}(\blds{\theta})\geq \frac{5\gamma}{P^2}$, $\sigma^2_S(\mathbf{A}_{\mathbb{U}_{\mathbb{S}}})$ scales as $cP^2$ (for some constant $c$), and this helps
us obtain reliable angle estimation, although the covariance estimates may be unreliable. 
\vspace{-0.2cm}
\section{Simulations}\label{sec:simulations}
We numerically investigate the useful SNR regime for coarray processing (\cref{sec:num_coarray_vs_direct}), the impact of SNR and the number of snapshots on DOA estimation error (\labelcref{sec:num_pr_res,sec:num_L_vs_SNR}), the relationship between DOA and covariance estimation error (\labelcref{sec:num_doa_vs_cov}), and the effect of the dynamic range of source powers on resolving two closely spaced sources (\labelcref{sec:num_dyrange}).
\subsection{When is Coarray-Based DOA Estimation Beneficial?}\label{sec:num_coarray_vs_direct}
We begin by examining under which circumstances coarray-based algorithms offer an advantage over more conventional DOA estimation methods. Specifically, in case of the ULA, we could apply MUSIC or ESPRIT \emph{directly} to the sample covariance matrix $\mathbf{\widehat{R}}_{\mf{y}}$ in \eqref{eqn:samp_cov} instead of the averaged coarray covariance matrix $\mf{\widehat{T}}_{ca}$ in \eqref{eqn:T_hat_est}. \cref{fig:direct_esprit} shows the matching distance error of coarray ESPRIT and direct ESPRIT, averaged over $10^3$ Monte Carlo trials, in case of the ULA, and, for comparison, coarray ESPRIT in case of the nested array with the same number of sensors ($P=20$). We consider $L=100$ snapshots, and $S=4$ equipower sources equally spaced by $\Delta = 2/P$. At medium to low SNR, the advantage of coarray-based processing is apparent. At high SNR, the situation is reversed, as the error of direct ESPRIT continues decreasing as a function of SNR, whereas the error of coarray ESPRIT saturates\footnote{This well-known and fundamental phenomenon is due to the finite-snapshot error of the coarray covariance matrix, see \cite{koochakzadeh2016cramerrao,wang2017,liu2017cramerrao}.}. However, coarray-based processing---including redundancy averaging \eqref{eqn:t_hat}---can clearly offer significant benefits in SNR or snapshot-limited conditions. As mostly such challenging scenarios are of interest in many applications, we focus on coarray ESPRIT herein.
\begin{figure}[h]
	\centering
	\includegraphics[width=0.74\linewidth]{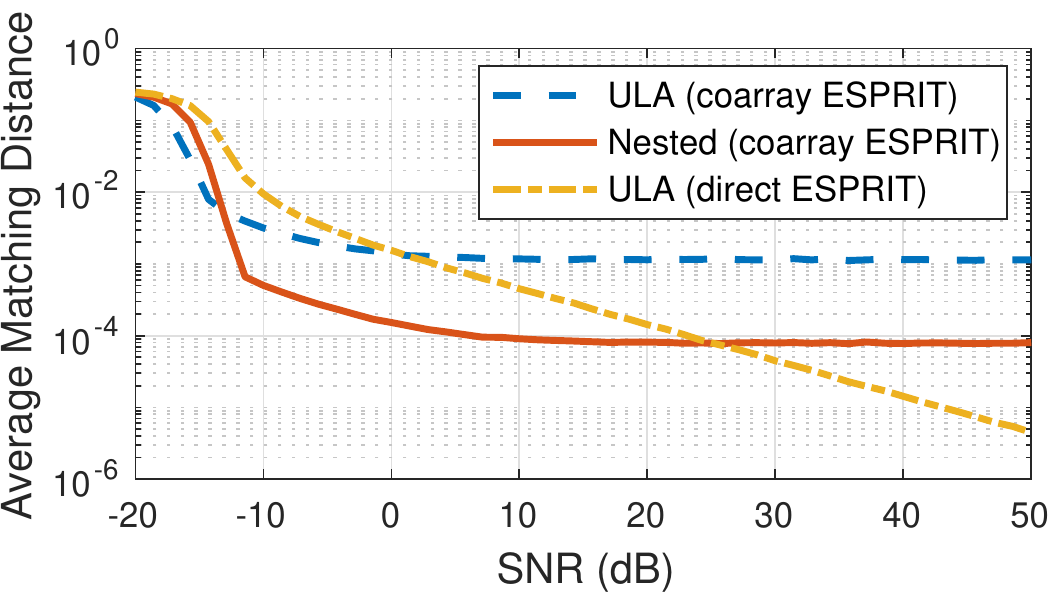}
	\caption{Comparison of ESPRIT applied to the sample covariance matrix \eqref{eqn:samp_cov} (direct ESPRIT) and the estimated coarray covariance matrix \eqref{eqn:T_hat_est} (coarray ESPRIT). Coarray ESPRIT achieves lower angle estimation error than direct ESPRIT at medium to low SNR.}
	\label{fig:direct_esprit}
\end{figure}
\vspace{-0.5cm}
\subsection{Improving Resolution by Increasing SNR or Snapshots}\label{sec:num_pr_res}
Next, we compare the probability of resolution as a function of the minimum separation for the nested array and ULA with the same number of sensors, $P=20$. Coarray ESPRIT is employed for both array geometries. We consider two sources with equal power ($p_1=p_2$) and (normalized) angles $\blds{\omega}=\{0.1, 0.1+\Delta\}$. The sources are declared to be successfully resolved when the estimated DOAs satisfy $\max_i \vert \hat{\omega}_i-\omega_i\vert \leq \Delta/10$. \cref{fig:prob_res} shows the empirical probability of resolution (averaged over $1000$ Monte-Carlo trials) for varying separation $\Delta$ and a fixed number of snapshots $L=55$ and SNR $=0$ and $-16$ dB. We observe that both array geometries can operate at a smaller separation at a higher SNR, i.e., smaller $\sigma/p_{\min}$ ratio. Indeed, the transition from low to high probability of resolution occur around $\Delta \propto 1/P$ for the ULA and $\Delta\propto 1/P^2$ for the nested array, as predicted by \cref{thm:well_sep_nst,thm:well_sep_ula}. It is also possible to enhance resolution by increasing the number of snapshots, as \cref{fig:prob_res} demonstrates. Here, the SNR is fixed at $0$ dB and the number of snapshots is $L=55$ and $L=600$, respectively.


\begin{figure}[!t]
	\centering
{\includegraphics[width=0.74\linewidth]{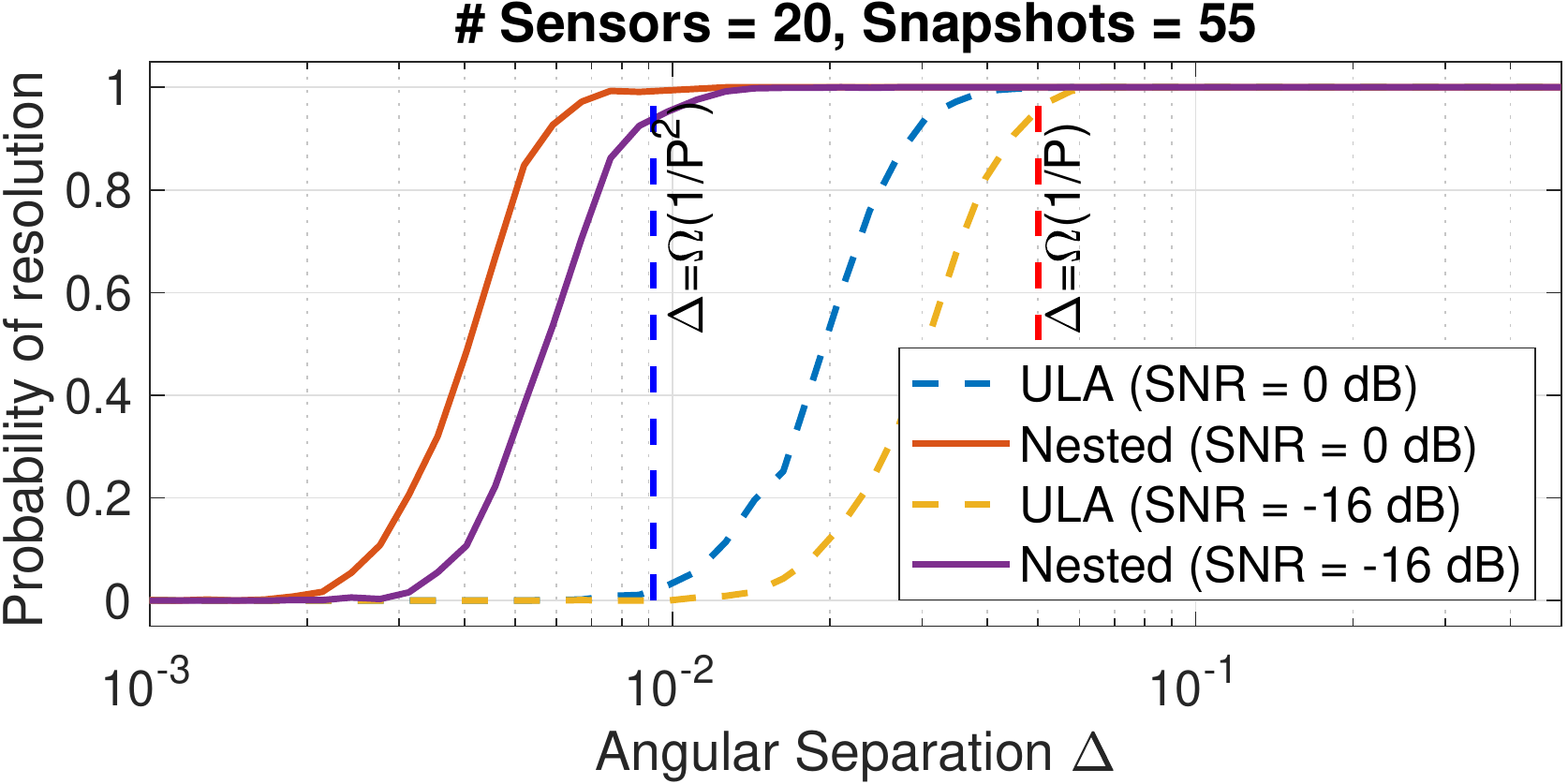}\label{fig:prob_res_snr}}\\
 {\includegraphics[width=0.74\linewidth]{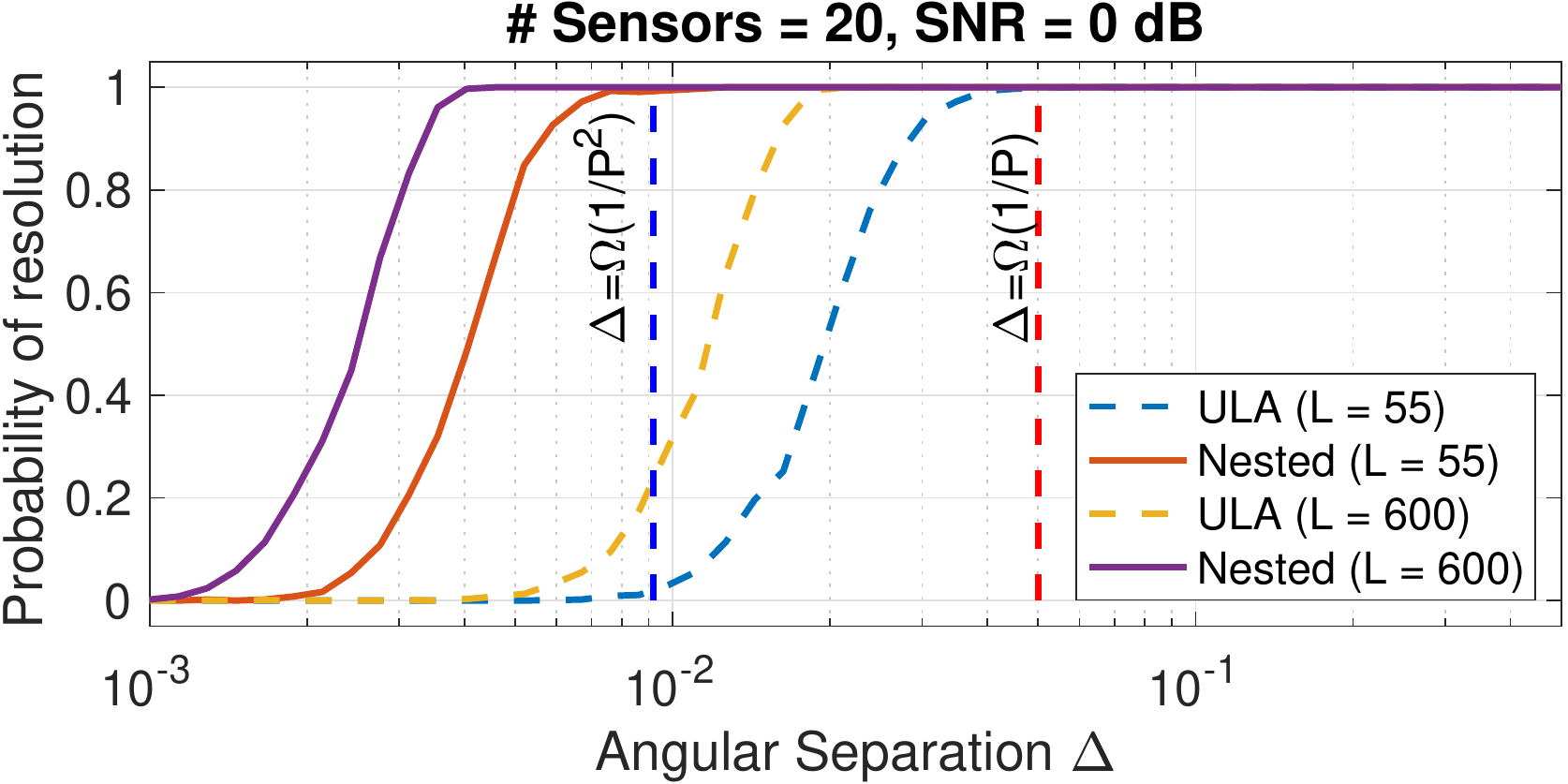}\label{fig:prob_res_snap}}
	\caption{Probability of resolution vs. source separation for different SNR levels (top) and number of snapshots (bottom). Increasing either improves resolution for both arrays.}
	\label{fig:prob_res}
\end{figure}

\subsection{Snapshot and SNR Trade-off} \label{sec:num_L_vs_SNR}
\cref{sec:num_pr_res} showed that SNR and the number of temporal snapshots can be exchanged for improved resolution. We now study this trade-off in further detail. We consider $S=2$ equipowered sources located at $\blds{\omega}=\{0.1,0.1+\Delta\}$, where $\Delta\in\{2/P,2/P^2\}$ and $P=20$. \cref{fig:SNR_L} shows the separation-relative matching distance error $\text{md}(\blds{\theta},\blds{\widehat{\theta}})/\Delta$ (averaged over $10^3$ Monte Carlo trials) as a function of both the number of snapshots and SNR. Firstly, fewer snapshots are required at higher SNR (and vice versa) to obtain the same recovery error, both in case of the ULA (left column) and nested array (right column). This supports \cref{thm:main_coesp}, where the matching distance depends on the number of snapshots and SNR through \labelcref{eqn:l_bound,eqn:beta}, respectively. Secondly, the nested array displays a more advantageous trade-off between snapshots and SNR compared to the ULA for both source separation $2/P$ (top row) and $2/P^2$ (bottom row). The benefit is especially apparent for $\Delta = 2/P^2$, where the nested array has a greatly larger range of operating points where the relative matching distance is low, as predicted by \cref{thm:well_sep_nst}. Note that the gray pixels correspond to a relative error of approximately $10\%$ of the separation, whereas white corresponds $\leq 1\%$ error.
\begin{figure}[!h]
    \centering
    \begin{tabular}{cc}\includegraphics[width=0.45\linewidth]{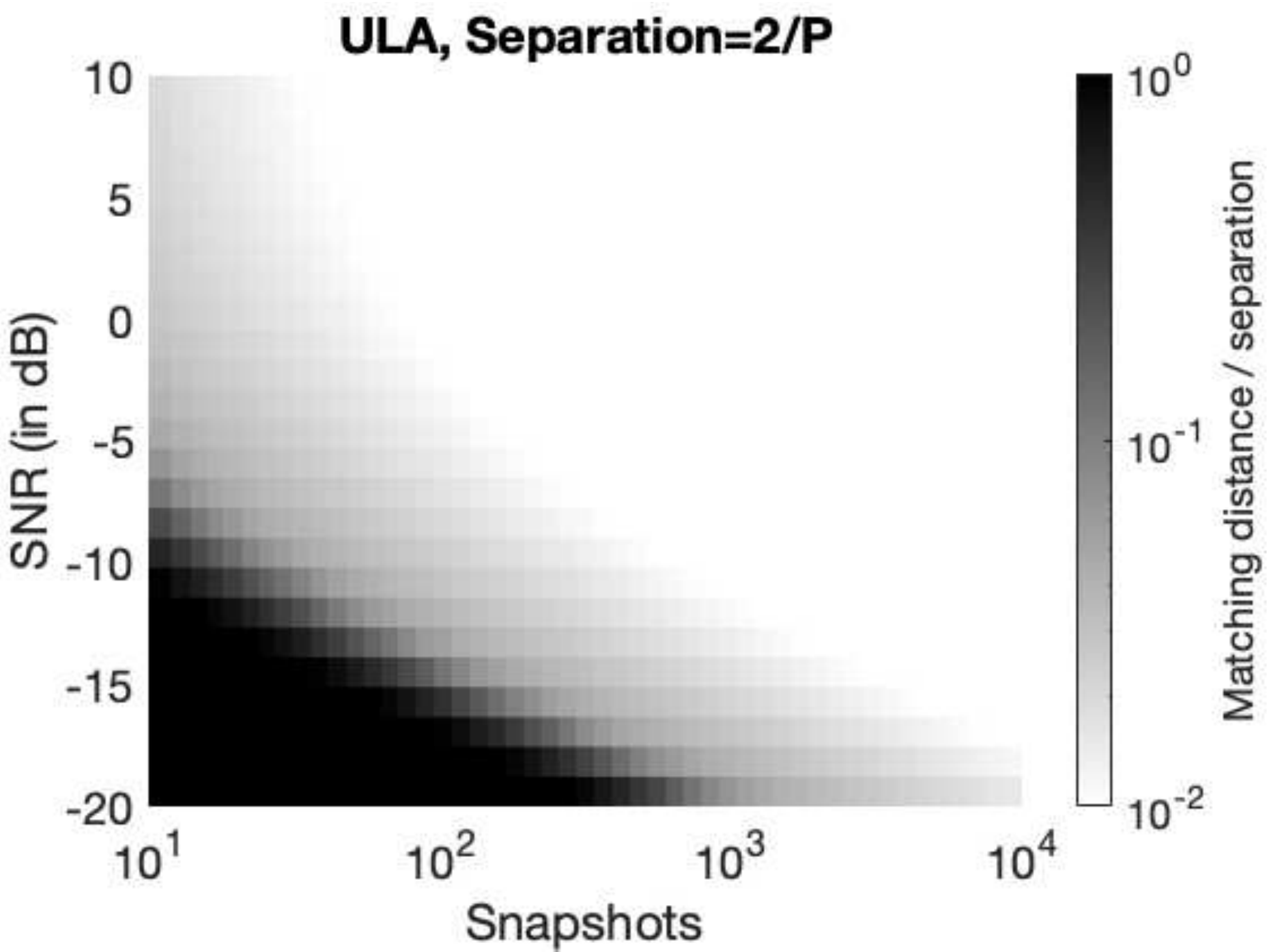}&
         \includegraphics[width=0.45\linewidth]{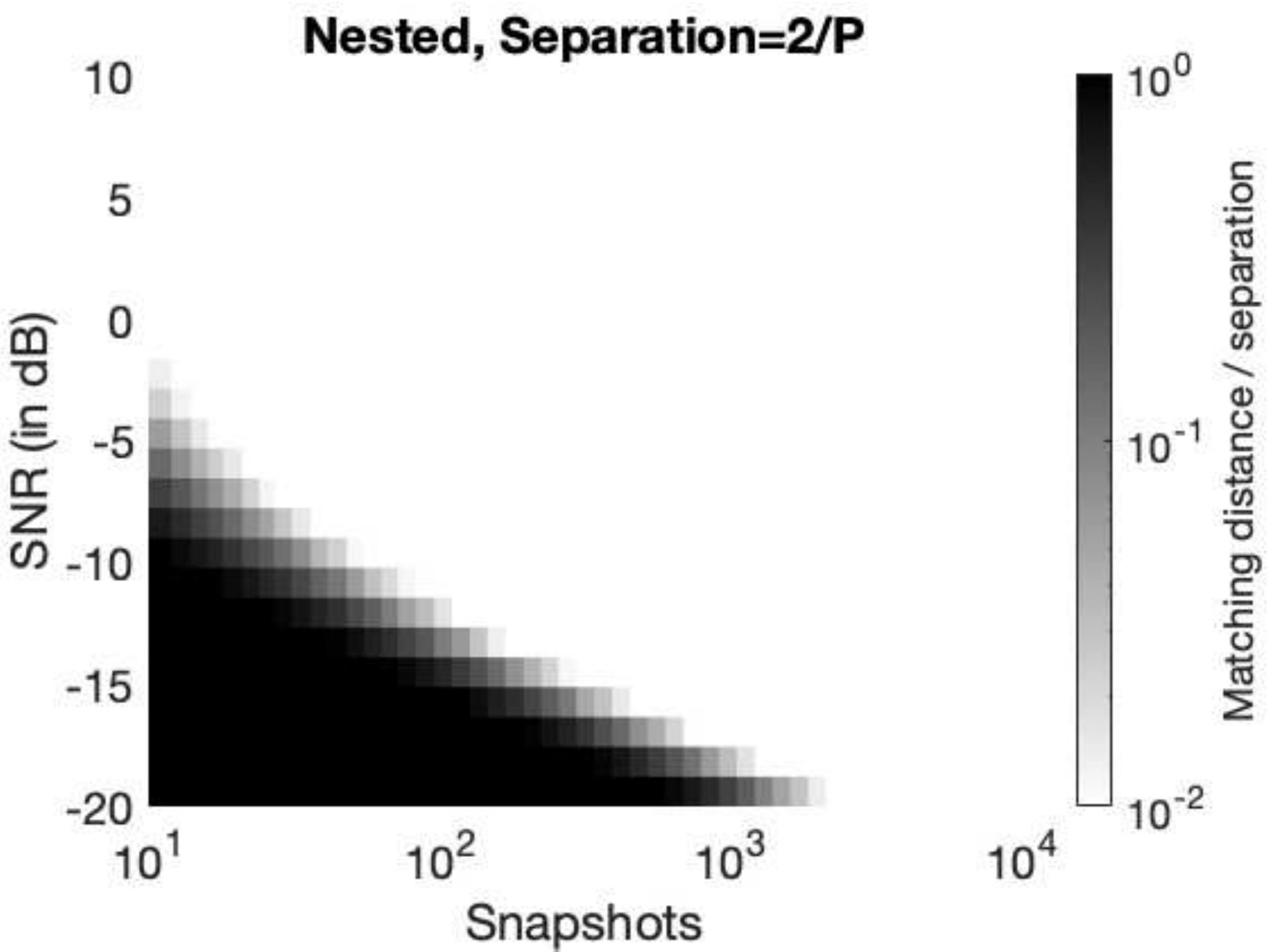}\\
        \includegraphics[width=0.45\linewidth]{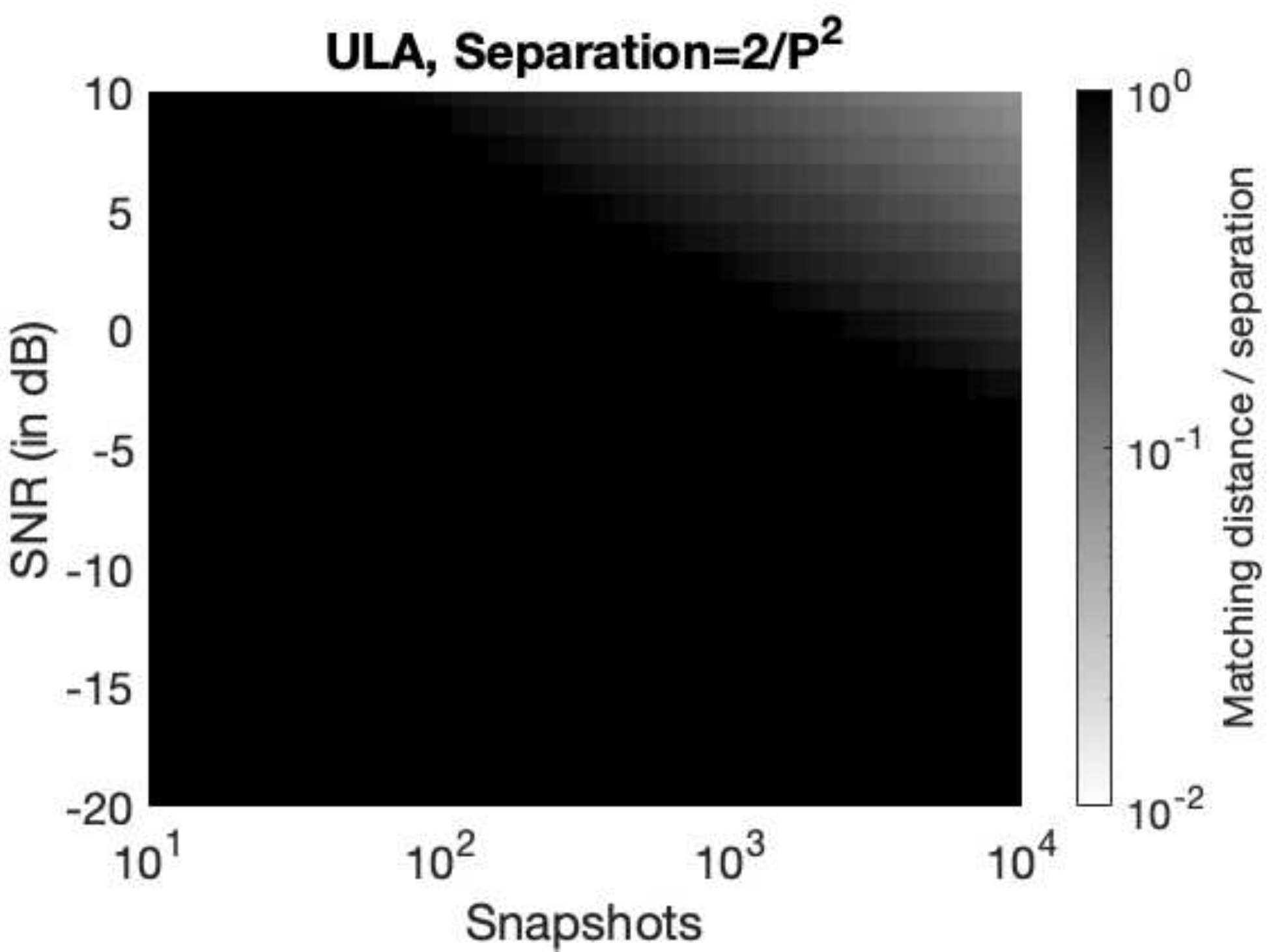}&
        \includegraphics[width=0.45\linewidth]{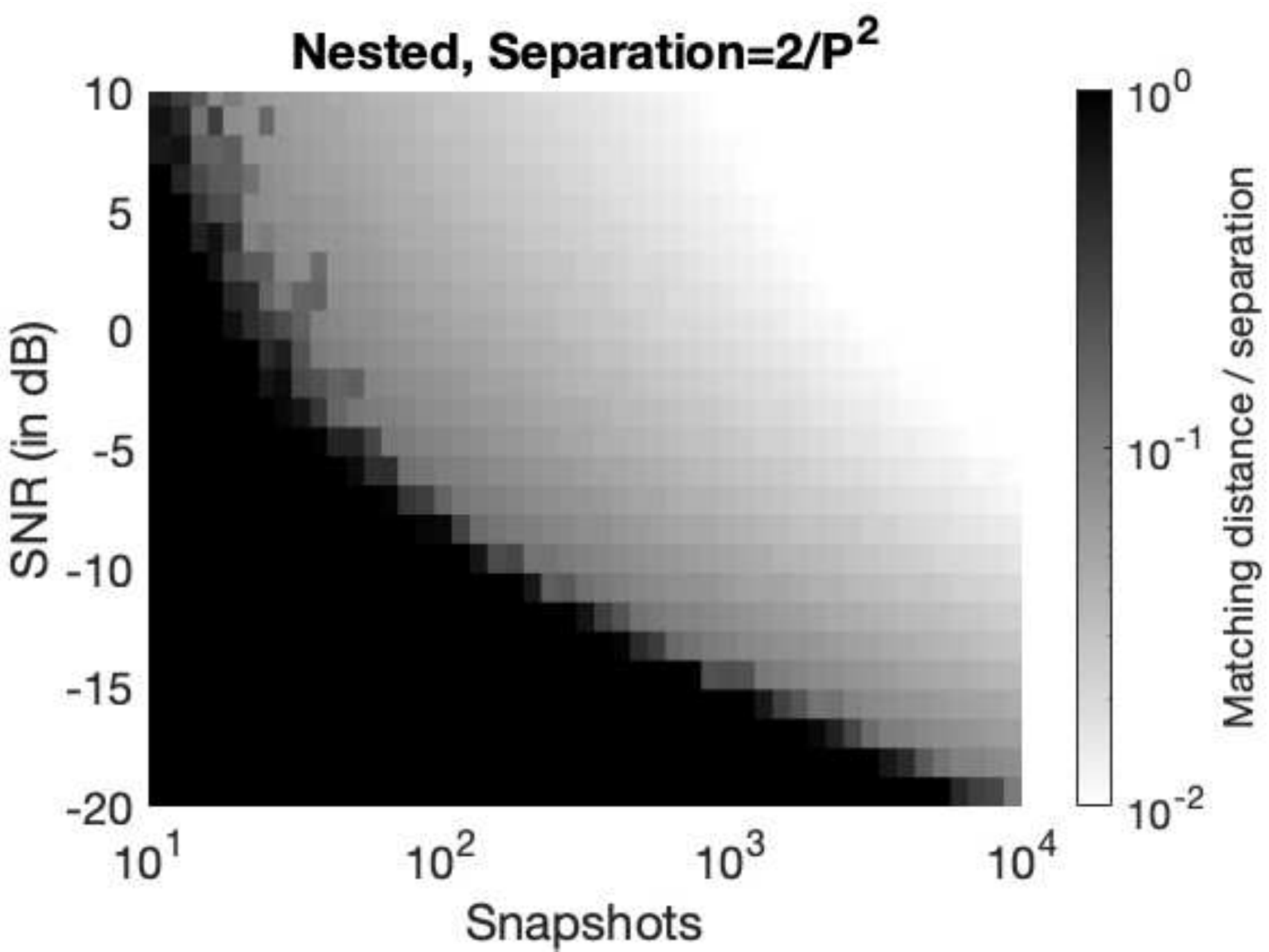}
    \end{tabular}
    \caption{Relative matching distance error $\text{md}(\blds{\theta},\blds{\widehat{\theta}})/\Delta$ as a function of snapshots and SNR. The nested array (right column) achieves lower error than the ULA (left column) for both source separation $\Delta=2/P$ (top row) and $\Delta=2/P^2$ (bottom row).}
    \label{fig:SNR_L}
\end{figure}
\vspace{-0.5cm}
\subsection{DOA and Covariance Estimation Error}\label{sec:num_doa_vs_cov}
Next, we illustrate an intriguing benefit of coarray-based DOA estimation in case of the nested array. We consider the average DOA matching distance and average covariance estimation error defined as {$\Vert \mf{T}_{ca}-\mf{\widehat{T}}_{ca}\Vert_2$}
for a varying number of sensors $P$ and $S=4$ equipower sources equally spaced by $\Delta\in \{1/P^{1.5},1/P^2\}$. The number of snapshots is $L=50$ and SNR $=0$ dB. \cref{fig:fix_snap_P} shows that the nested array incurs a larger covariance estimation error compared to the ULA with the same number of sensors. However, despite obtaining a worse estimate of the covariance matrix $\mf{\widehat{T}}_{ca}$, the nested array achieves superior DOA estimation performance when coarray ESPRIT is applied to $\mf{\widehat{T}}_{ca}$. In fact, when the separation is $\Delta=1/P^2$, the average matching distance no longer decays with $P$ for the ULA, whereas it continues to do so for the nested array. This is enabled by the larger coarray aperture of the nested array, which offsets the effect of finite snapshot covariance estimation error as discussed in \cref{sec:myth}. Note that for a fixed number of snapshots and a growing number of sensors $P$, the entries of the coarray covariance matrix $\mf{T}_{ca}$ become increasingly challenging to estimate, since the size of $\mf{T}_{ca}$ is proportional to the number of coarray elements $M_{ca}$, which is $\propto P$ for the ULA and $\propto P^2$ for the nested array.
\begin{figure}[!h]
	\centering
{\includegraphics[width=0.74\linewidth]{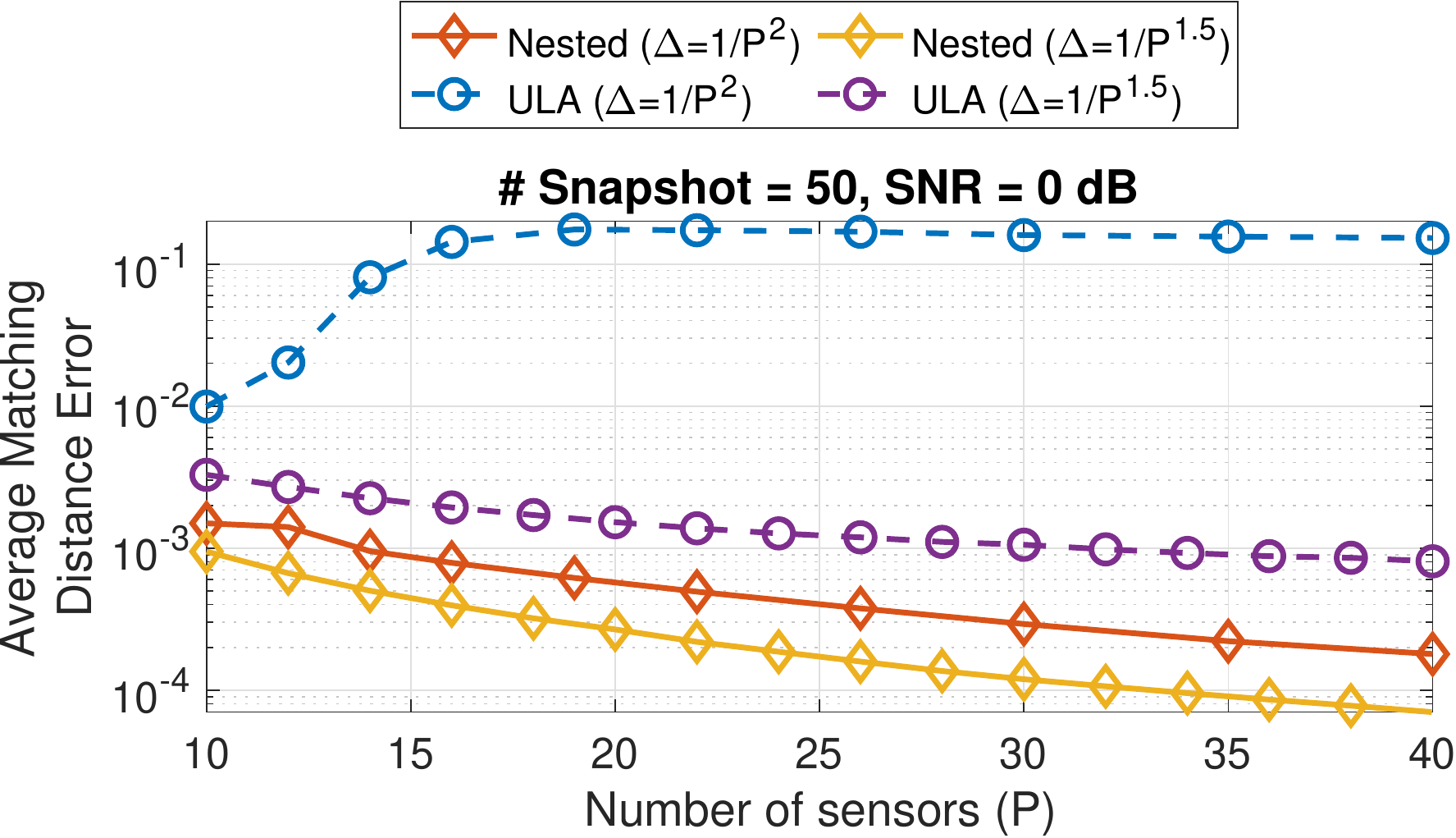}\label{fig:fix_snap_P_md}}\\
 {\includegraphics[width=0.74\linewidth]{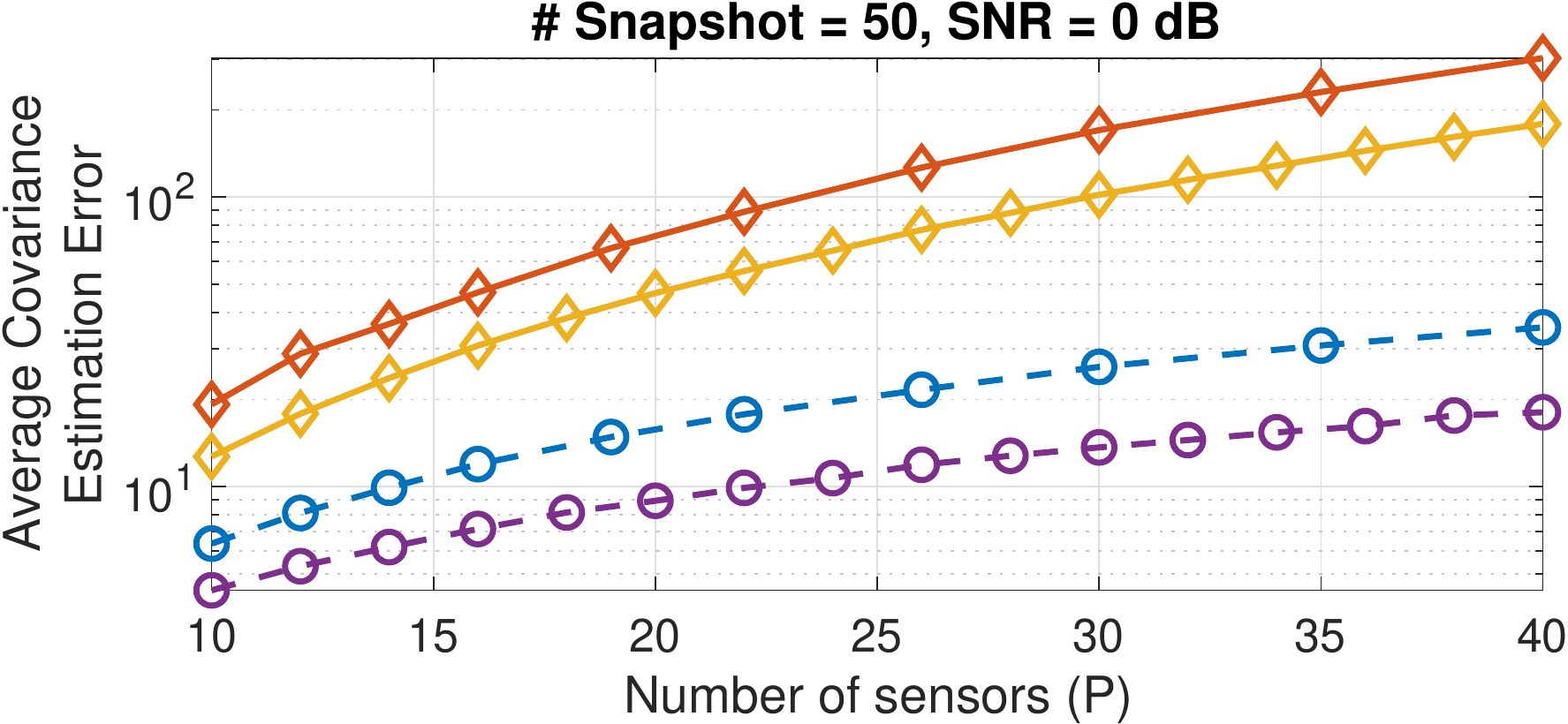}\label{fig:fix_snap_P_cov}}
	\caption{Average matching distance (top) and covariance estimation error (bottom) as a function of the number of sensors $P$. The DOA estimation error of the nested array decays despite the larger covariance estimation error compared to the ULA.}
	\label{fig:fix_snap_P}
\end{figure}
\vspace{-0.4cm}
\subsection{Effect of Dynamic Range of Source Powers} \label{sec:num_dyrange}
In the final experiment, we investigate the ability of coarray ESPRIT to resolve two sources with unequal powers. We set the dynamic range to $p_{\max}/p_{\min}\in \{1,10\}$ by fixing the power of the weaker source to $p_{\min} = 0.2$ and varying $p_{\max}$. \cref{fig:dyrange} shows that the number of snapshots required to distinguish two sources (separated by $\Delta=1/P$) is significantly larger when $p_{\max}/p_{\min}=10$ compared to $p_{\max}/p_{\min}=1$. {This is consistent with \cref{thm:well_sep_ula,thm:well_sep_nst}, which imply that the sufficient number of snapshots for resolving two sources (with high probability) grows with $p_{\max}$ if $p_{\min}$ and $\sigma$ are held fixed, irrespective of the array geometry. This brings out a non-trivial dependence of the dynamic range $p_{\max}/p_{\min}$ on the sample complexity.} Hence, distinguishing two sources with greatly different powers is more challenging and requires more snapshots than when the powers are equal.
\begin{figure}[!h]
\setlength\belowcaptionskip{-1.5\baselineskip}
    \centering
    \includegraphics[width=0.78\linewidth]{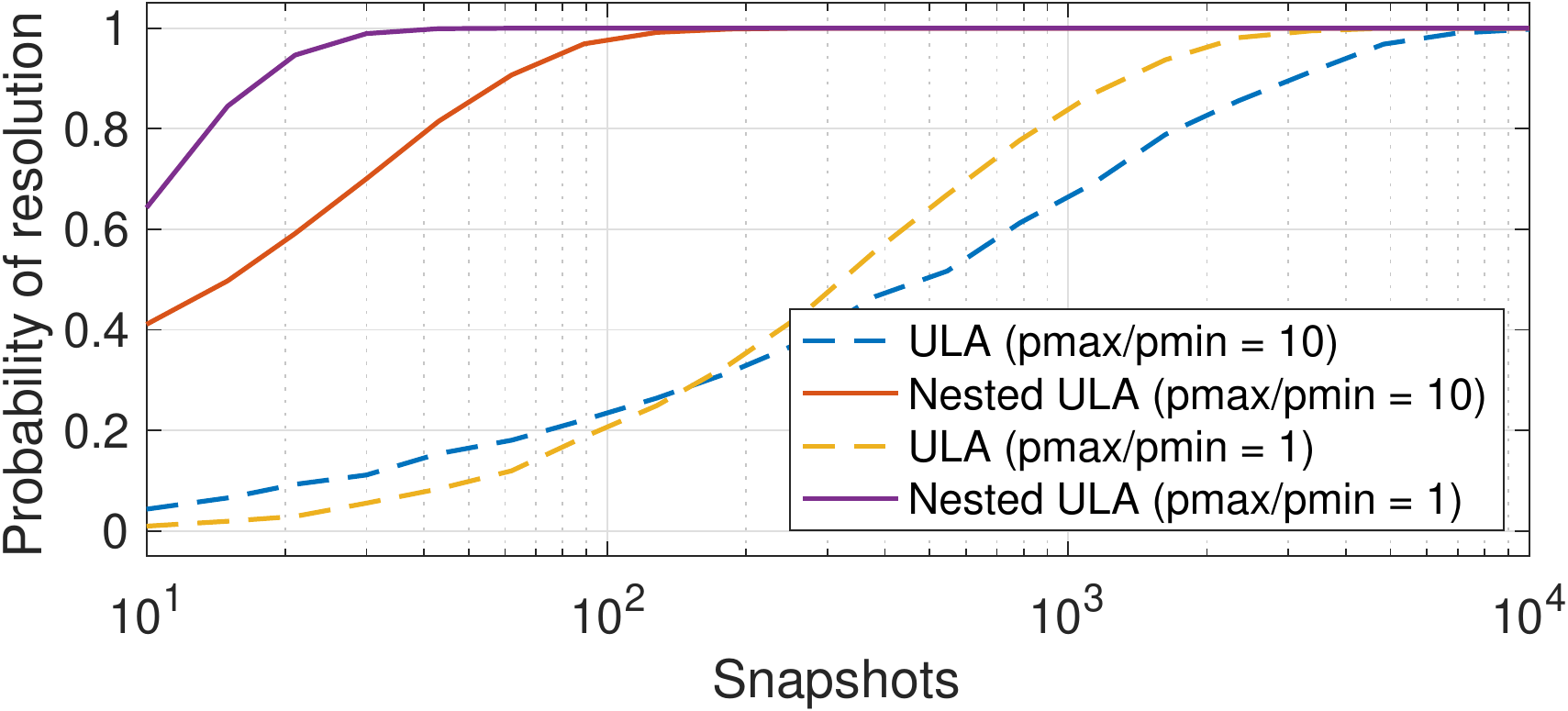}
    \caption{Effect of dynamic range of source powers on probability of resolution. Coarray ESPRIT requires more snapshots to detect two sources with larger dynamic range $p_{\max}/p_{\min}$.}
    \label{fig:dyrange}
\end{figure}
\vspace{-0.5cm}
\section{Conclusion} \label{sec:conclusion}
This paper investigated angle estimation error of coarray ESPRIT. We considered both additive noise and finite-snapshot covariance estimation error, which we probabilistically characterized in the case of Toeplitz covariance matrices. Our results show that if the number temporal snapshots scales logarithmically with the number of sensors, coarray ESPRIT achieves arbitrarily low estimation error with high probability. This also shows that the DOA estimation error can be small even though the covariance estimation error may be large. Finally, our theoretical and simulation results demonstrate that sparse arrays can provide higher resolution and better noise resilience compared to the ULA with the same number of sensors and snapshots.

\vspace{-0.2cm}
\begin{appendices}
\section{}
\subsection{Intermediate Results} \label{app:cov_est}
We will first state the complex extension of Hanson-Wright inequality\cite{hanson1971bound}, which is obtained by applying \cite[Theorem 1.1]{rudelson2013hanson} with the strategy described on \cite[Section 3.1, Page 9]{rudelson2013hanson}.
\begin{lem}\label{lem:hw_cplx}
Let $\mf{A} \in \mathbb{C}^{n \times n}$ be a fixed Hermitian matrix. Consider the random vector $\mf{x}=[x_1,x_2,\cdots,x_n]^{\top}\in \mathbb{C}^n$ with independent real and imaginary components $\text{Re}(x_i)$, $\text{Im}(x_i)$ satisfying 
$\E(\text{Re}(x_i))=\E(\text{Im}(x_i))=0$, and $\Vert \text{Re}(x_i)\Vert_{\psi_2}\leq K$, $\Vert \text{Im}(x_i)\Vert_{\psi_2}\leq K$. Then for any $\epsilon>0$, we have 
{\footnotesize
\begin{align*}
    &\mathbb{P}(\vert\mf{x}^H\mf{A}\mf{x}-\E(\mf{x}^H\mf{A}\mf{x})\vert\!>\!\epsilon)\leq 2 \exp\bigg(-c\min (\frac{\epsilon^2}{2K^4 \Vert \mf{A}\Vert_F^2},\frac{\epsilon}{K^2 \Vert \mf{A}\Vert_2})\bigg)
\end{align*}}
where $c>0$ is a universal constant.
\end{lem}
\begin{proof}
Let $\mf{z}=[\text{Re}(\mf{x})^{\top},\text{Im}(\mf{x})^{\top}]^{\top}\in \mathbb{R}^{2n}$ and define :
$
    \mf{\tilde{A}}=\begin{bmatrix}
    \text{Re}(\mf{A}) & -\text{Im}(\mf{A})\\
    \text{Im}(\mf{A}) & \text{Re}(\mf{A})
    \end{bmatrix}
$. It is easy to see that for any Hermitian $\mf{A}$, we have the following equality
$    \mf{x}^{H}\mf{A}\mf{x}=\mf{z}^{T}\mf{\tilde{A}}\mf{z}
$. Further, it can be verified that $\Vert\mf{\tilde{A}}\Vert_F=\sqrt{2}\Vert\mf{A}\Vert_F$ and $\Vert\mf{\tilde{A}}\Vert_2=\Vert\mf{A}\Vert_2$. Now, we can apply \cite[Theorem 1.1]{rudelson2013hanson}, to obtain the desired probability bound. 
\end{proof}
\vspace{-0.2cm}
\begin{lem}\label{lem:hw_snap}
Let $\mathbf{w}_i \in \mathbb{C}^n$, $1 \leq i \leq T$ be i.i.d complex circularly symmetric Gaussian random variable with distribution $\mathcal{CN}(\mathbf{0},\mathbf{\Sigma})$. Let $\mathbf{A} \in \mathbb{C}^{n \times n}$ be a fixed Hermitian matrix, then for any $\epsilon>0$ and universal constant $c$, we have
$P( \vert\frac{1}{L}\sum_{i =1}^{L} \mathbf{w}^H_{i}\mathbf{A}\mathbf{w}_{i} - \E[\mathbf{w}^H_{i}\mathbf{A}\mathbf{w}_{i} ]   \vert \geq \epsilon)
\leq  2\exp\left(-cL \min\left(\frac{\epsilon^2}{2K^4\Vert \mathbf{\Sigma}\Vert_2^2 \Vert \mathbf{A} \Vert_F^2},\frac{\epsilon}{K^2\Vert \mathbf{\Sigma}\Vert_2 \Vert \mathbf{A} \Vert_2}\right)\right)$.
\end{lem}
\vspace{-0.2cm}
\begin{proof}
Since $\mathbf{w}_i$ is a complex circularly symmetric Gaussian random variable distributed according to $\mathcal{CN}(\mathbf{0},\mathbf{\Sigma})$, we define a new transformed variable $\smash{\mathbf{u}_i=\mathbf{\Sigma}^{-1/2}\mf{w}_i}$ where $\mathbf{\Sigma}^{1/2}$ is the square root of the covariance matrix $\mathbf{\Sigma}$. It can be verified that $\smash{\mf{u}_i\sim \mathcal{CN}(\mf{0},\mf{I}_n)}$, i.e., it is also a complex circularly symmetric Gaussian random variable with independent real and imaginary components. Define block-wise diagonal matrices
$\Tilde{\mathbf{A}} = \text{diag}(\mathbf{A},\hdots,\mathbf{A}),\Tilde{\mathbf{\Sigma}}^{1/2} = \text{diag}(\mathbf{\Sigma}^{1/2},\hdots,\mathbf{\Sigma}^{1/2})\in \mathbb{C}^{nL \times nL}$
and $\Tilde{\mathbf{u}} = [\mathbf{u}^T_{1}, \hdots,\mathbf{u}^T_{L}]^T \in \mathbb{C}^{nL}$. 
Next, we can observe that 
$
\sum_{i = 1}^{L} \mathbf{w}_i^H \mathbf{A} \mathbf{w}_i = \sum_{i = 1}^{L} \mathbf{u}_i^H \mathbf{\Sigma}^{1/2} \mathbf{A} \mathbf{\Sigma}^{1/2} \mathbf{u}_i= \mathbf{\tilde{u}}^H \mathbf{\tilde{\Sigma}}^{1/2} \mathbf{\tilde{A}}\mathbf{\tilde{\Sigma}}^{1/2}  \mathbf{\tilde{u}} 
$.
We have $\mathbb{E} (\sum_{i = 1}^{L} \mathbf{w}_i^H \mathbf{A} \mathbf{w}_i)=L \mathbb{E} \left(\mf{w}_i^H\mf{A}\mf{w}\right)
$, since it is a sum of $L$ i.i.d random variables.
The desired probability can be re-written as:
$P( \vert\frac{1}{L}\sum_{i =1}^{L} \text{Re}(\mathbf{w}^H_{i}\mathbf{A}\mathbf{w}_{i}) - \mathbb{E}[\text{Re}( \mathbf{w}^H_{i}\mathbf{A}\mathbf{w}_{i}) ]   \vert \geq \epsilon)
= P( \vert \text{Re}(\Tilde{\mathbf{u}}^H \Tilde{\mathbf{\Sigma}}^{1/2}\Tilde{\mathbf{A}}\Tilde{\mathbf{\Sigma}}^{1/2}\Tilde{\mathbf{u}})-\mathbb{E}[\text{Re}(\Tilde{\mathbf{u}}^H \Tilde{\mathbf{\Sigma}}^{1/2}\Tilde{\mathbf{A}}\Tilde{\mathbf{\Sigma}}^{1/2}\Tilde{\mathbf{u}})] \vert \geq L\epsilon)$. 
Recall that $\text{Re}(\tilde{u}_i), \text{ Im}(\tilde{u}_i)$ are i.i.d distributed as $\mathcal{N}(0,1/2)$ and hence sub-Gaussian with $K=2/\sqrt{3}$. Note that due to the block-diagonal structure we have
$\smash{
\Vert\Tilde{\mathbf{\Sigma}}^{1/2}\Tilde{\mathbf{A}}\Tilde{\mathbf{\Sigma}}^{1/2}\Vert_F^2=L \Vert\mathbf{\Sigma}^{1/2}\mathbf{A}\mathbf{\Sigma}^{1/2}\Vert_F^2 \leq L \Vert \mf{A}\Vert_F^2 \Vert\blds{\Sigma} \Vert_2^2}
$
and
$\smash{
\Vert\Tilde{\mathbf{\Sigma}}^{1/2}\Tilde{\mathbf{A}}\Tilde{\mathbf{\Sigma}}^{1/2}\Vert_2=\Vert\mathbf{\Sigma}^{1/2}\mathbf{A}\mathbf{\Sigma}^{1/2}\Vert_2\leq \Vert \mf{A}\Vert_2\Vert\mf{\Sigma}\Vert_2}$. The proof is completed by applying \cref{lem:hw_cplx} with $\epsilon=\epsilon L$. 
\end{proof}
\vspace{-0.6cm}
\subsection{Proof of Theorem \ref{thm:coverr}}\label{app:thm1}
From Lemma \ref{lem:spectral_norm}, we have $P(\Vert \mf{E}_L \Vert_2\geq \epsilon)\!\leq\! P( \sup\vert f_{\mf{e}}(\theta)\vert \!\geq\! \epsilon )$. In general, it is not straightforward to evaluate this supremum, however, we exploit the following result from  \cite{zygmund2002trigonometric}that bounds it by using the function value evaluated at a few grid points. 
\begin{lem} \cite[Theorem 7.28, Chapter 10, Vol.2, Pg. 33]{zygmund2002trigonometric} \label{lem:finite_point}Let $f(\theta)$ be a trigonometric polynomial of order $N$. Then,
$
\sup_{\theta \in [-\pi, \pi]} |f(\theta)| \le  
2 \max_{1 \le k \le 4N} |f(\theta_k)|, \quad \theta_k = \frac{k - 2N}{4N} \pi. 
$
\end{lem}
From \cref{prop:spec_func}, we have $f_{\mf{e}}(\theta)=\text{tr}(\Ey\Lamth)$. However, we want to relate it to the sample covariance matrix $\Ryhat$. In order to do this, we show that $\text{tr}(\Ryov\Lamth)=\text{tr}(\Ryhat\Lamth)$ where recall from \eqref{eqn:samp_cov} that $\Ryhat$ is the sample covariance matrix: 
{\footnotesize
\begin{align*}
   &\text{tr}(\Ryov\Lamth)=\sum_{m=1}^{P}\sum_{n=1}^{P}[\Ryov]_{m,n}[\Lamth]_{n,m}\\
    &=\sum_{s=-M_{ca}}^{M_{ca}}\sum_{m,n: \atop d_m-d_n=s} \hat{t}_s \frac{\exp(-js\theta)}{\vert \Omega_s\vert}=\sum_{s=-M_{ca}}^{M_{ca}}\hat{t}_s\vert \Omega_s\vert \frac{\exp(-js\theta)}{\vert \Omega_s\vert}\\
    &\underset{(a)}{=}\sum_{s=-M_{ca}}^{M_{ca}}\sum_{m,n: \atop d_m-d_n=s}[\Ryhat]_{m,n}[\Lamth]_{n,m}=\text{Tr}(\Ryhat\Lamth),
\end{align*}}
where $(a)$ follows from the redundancy averaged estimator where for all $m,n$ such that $d_m-d_n=s$, we have $\vert\Omega_{s}\vert\hat{t}_{s}=\sum_{d_m-d_n=s}[\Ryhat]_{m,n}$. Therefore, we have the following relation:
$f_\mf{e}(\theta)\!=\!\text{tr}\left(\Ey\Lamth\right)\!=\!\text{tr}\left((\mf{R}_{y}-\Ryov)\mf{\Lambda}(\theta)\right)=\text{tr}((\mf{R}_{y}-\mf{\widehat{R}}_{y})\mf{\Lambda}(\theta))\!=\!\frac{1}{L}\sum_{t=1}^{L}\left(\E[\mf{y}(t)^H\mf{\Lambda}(\theta)\mf{y}(t)]-\mf{y}(t)^H\mf{\Lambda}(\theta)\mf{y}(t)\right)$. 
Since the snapshots are i.i.d, we can define i.i.d random variables $\{ Z_t(\theta)\}_{t=1}^{L}$ as 
$
Z_t (\theta)  \triangleq \mathbf{y}(t)^H \mathbf{\Lambda}(\theta) \mathbf{y}(t) - \mathbb{E} (\mathbf{y}(t)^H \mathbf{\Lambda}(\theta) \mathbf{y}(t))
$ with $\mf{y}(t)\sim \mathcal{CN}(\mf{0},\Ry)$.
Note that $\Lamth$ is Hermitian.
Hence, we can apply Lemma \ref{lem:hw_snap} with $\mf{\Sigma}=\Ry$ and $\mf{A}=\Lamth$ to obtain $\forall \epsilon > 0$,
{\footnotesize\ben
&&\hspace{-0.3in}\mathbb{P}\left( \frac{1}{L} |\sum_{t = 1}^{L} Z_t(\theta)|\ge \epsilon \right) \le \label{eqn_bound_sum_Z_l}\\
&&\hspace{-0.3in} 2 \exp \left[ -c L \min \left( \frac{\epsilon^2}{2K^4\| \mathbf{R_{y}} \|_2^2 \|\mathbf{\Lambda}(\theta) \|_F^2}, \frac{\epsilon}{K^2\| \mathbf{R_{y}} \|_2 \| \mathbf{\Lambda}(\theta) \|_2}\right) \right]. \nonumber 
\een}
We want to obtain a universal upper bound that is similar to \eqref{eqn_bound_sum_Z_l} but not dependent on $\theta$. Notice,
$
\Vert\Lambda(\theta)\Vert_F^2=\frac{1}{\vert\Omega_0\vert}+\sum_{s=1}^{M_{ca}}\frac{2}{\vert \Omega_s\vert}\leq 2\Delta(\mathbb{S})
$.
Similarly, we can also bound $\Vert\Lambda(\theta)\Vert_2\leq \Vert\Lambda(\theta)\Vert_F \leq \sqrt{2\Delta(\mathbb{S})}$. This gives us the following bound:
{\footnotesize
\begin{align}
&\mathbb{P}\left( \frac{1}{L} |\sum_{t = 1}^{L} Z_t(\theta)|\ge \epsilon \right) \le \label{eqn:theta_free}\\
&\hspace{-0.2cm}2 \exp \left[ -c L \min \left( \frac{\epsilon^2}{4K^4\| \mathbf{R_{y}} \|_2^2 \Delta(\mathbb{S})},\frac{\epsilon}{K^2\| \mathbf{R_{y}} \|_2 \sqrt{2\Delta(\mathbb{S})}}\right) \right]. \nonumber 
\end{align}}
Note $f_{\mf{e}}$ is a trigonometric polynomial of order $M_{ca}$. Now, we will use Lemma \ref{lem:finite_point} to bound 
the spectral function $\vert f_{\mf{e}}(\theta)\vert$. 
$\mathbb{P} (\sup_{\theta \in [-\pi,\pi]}\vert f_{\mf{e}}(\theta)\vert\!\ge\! \epsilon)\!\le\!\mathbb{P} (2 \max_{1 \le k \le 4M_{ca}} |f_{e}(\theta_k)|\!\ge\!\epsilon)\le \sum_{k=1}^{4M_{ca}} \mathbb{P} \left( |f_{\mf{e}}(\theta_k)| \ge \frac{\epsilon}{2} \right)
\le 8 M_{ca} \exp \big[ -c_1 L \min \big( \frac{c_2\epsilon^2}{\| \mathbf{R_{y}} \|_2^2 \Delta(\mathbb{S})}, \frac{\epsilon}{\| \mathbf{R_{y}} \|_2 \sqrt{\Delta(\mathbb{S})}}\big) \big]$, 
where $c_1=c/(2\sqrt{2}K^2)$ ($c$ was given in \cref{lem:hw_snap}) and $c_2=1/(4\sqrt{2}K^2)=3/(16\sqrt{2})<1$. The first inequality follows due to Lemma \ref{lem:finite_point}, the second inequality follows from union bound. The last inequality is a consequence of the bound computed in \eqref{eqn:theta_free}.
\section{}
\subsection{Proof of Theorem \ref{Espiriterror}}
\label{app:ESP}
The proof uses several results from \cite{WeilinEsprit}. However, unlike \cite{WeilinEsprit} the underlying subspace of interest is the coarray subspace and the perturbation is due to covariance estimation error and noise. We provide key intermediate steps to make the results self-contained. 

Recall that columns of $\mf{U}$ and $\mf{\widehat{U}}$ are orthonormal bases for the subspaces $\mathcal{R}(\mf{U})$ and $\mathcal{R}(\mf{\widehat{U}})$. Let the principal angles between the subspaces  $\mathcal{R}(\mf{U})$ and $\mathcal{R}(\mf{\widehat{U}})$ be denoted as $\boldsymbol{\Theta}( \mathcal{R}(\mf{U}),\mathcal{R}(\mf{\widehat{U}})):=[\psi_1,\psi_2,\cdots,\psi_S]^T$ where $0 \leq \psi_1 \leq \psi_2 \leq \cdots \leq \psi_S \leq \pi/2$. Then from \cite{stewart1990matrix}, we have $ \cos(\psi_i)=\sigma_i( \mf{U}^H\mf{\widehat{U}})\ i=1,2,\cdots,S$.
Recall from \cref{lem:ESP_inv}, the output of ESPRIT is invariant to the choice of the basis. 
For ease of analysis, we will choose a pair of basis for $\mathcal{R}(\mf{U})$ and $\mathcal{R}(\mf{\widehat{U}})$, which are also known as ``canonical bases" \cite{WeilinEsprit}.  Let the SVD of the matrix $\mf{U}^{H}\mf{\widehat{U}}$ be of the form $\mf{U}^{H}\mf{\widehat{U}}\!:=\!\mf{L}\mf{\Sigma}_c\mf{R}^H, \mf{L},\mf{R}\in \mathbb{C}^{S\times S}$,
where $\mf{\Sigma}_c=\text{diag}(\sigma^{c}_1,\sigma^{c}_2,\cdots,\sigma^{c}_{S})$ where $\sigma_i^{c}=\sigma_i(\mf{U}^H\mf{\widehat{U}})$ are arranged in descending order. The canonical basis $\mf{U}^{(\text{c})}$ and $\mf{\widehat{U}}^{(\text{c})}$ are given by:
\begin{align}
  \mf{U}^{(\text{c})}:= \mf{U}\mf{L}, \quad\mf{\widehat{U}}^{(\text{c})}:=\mf{\widehat{U}}\mf{R} \label{eqn:u_can}
\end{align}
Using the canonical basis, we define the following matrices:
$
    \mf{\Psi}^{(\text{c})}\!:=\! \mf{U}^{(\text{c})\dagger}_0\mf{U}^{(\text{c})}_1,\ \mf{\widehat{\Psi}}^{(\text{c})}\!:=\! \mf{\widehat{U}}^{(\text{c})\dagger}_0\mf{\widehat{U}}^{(\text{c})}_1$. 
Since $\mathcal{R}(\mf{U})\!=\!\mathcal{R}(\mf{U}^{(\text{c})})$ and $\mathcal{R}(\mf{\widehat{U}})\!=\!\mathcal{R}(\mf{\widehat{U}}^{(\text{c})})$, we have $
\boldsymbol{\Theta}(\mathcal{R}(\mf{U}^{(\text{c})}),\mathcal{R}(\mf{\widehat{U}}^{(\text{c})}))\!=\!\boldsymbol{\Theta}(\mathcal{R}(\mf{U}),\mathcal{R}(\mf{\widehat{U}}))$. Notice that the canonical basis has the following property:
$\cos(\psi_i)=\sigma_i(\mf{U}^{\text{(c)}H}\mf{\widehat{U}}^{(\text{c})})=\mf{u}^{(\text{c})H}_i\mf{\widehat{u}}^{(\text{c})}_i$. 
We will use \cite[Lemma 2]{WeilinEsprit} that relates the matching distance error to the quantity $\Vert \mathbf{\widehat{\Psi}}^{(\text{c})} -\mathbf{\Psi}^{(\text{c})} \Vert_2$ and holds universally:
\ben 
\smash{\text{md}(\blds{\theta},\hat{\blds{\theta}}) \leq\pi \frac{S^{3/2}\sqrt{M_{ca}+1}}{\sigma_S(\mathbf{A}_{\mathbb{U}_{\mathbb{S}}}(\blds{\theta}))}\Vert \mathbf{\widehat{\Psi}}^{(\text{c})} -\mathbf{\Psi}^{(\text{c})} \Vert_2. \label{eqn:md_psi_bound}}
\een
\subsection{Relating $\Vert \mathbf{\widehat{\Psi}}^{(\text{c})} -\mathbf{\Psi}^{(\text{c})} \Vert_2$ to $\Vert\EL\Vert_2$} 
Let $\mf{B}=\mf{A}+\mf{N} \in \mathbb{C}^{M \times N}$, where $\text{rank}(\mf{A})\geq L$. Suppose $\psi_L$ is the largest principal angle between 
the subspace spanned by $L$ principal singular vectors (corresponding to $L$ largest singular values) of $\mf{A}$ and $\mf{B}$, respectively. If $\sigma_{L+1}(\mf{A})\leq \alpha$ and $\sigma_L(\mf{B})\geq \alpha+\delta$ for some $\alpha\geq 0$ and $\delta>0$ then, Wedin's Theorem \cite{wedin1972perturbation} states that:
\begin{equation}
   \smash{\sin(\psi_L) \leq {\|\mf{N}\|_2}/{\delta}}\label{eqn:wedin_main}
\end{equation}
\begin{lem}\label{anglebetweensubspace}
Suppose $ \sigma_S(\mathbf{T}_{ca}) \geq 2\Vert \mathbf{E}_L \Vert_2 $ and $\beta=p_{\min}\sigma_S^2(\mf{A}_{\mathbb{U}_{\mathbb{S}}}(\blds{\theta}))-\sigma^2>0$. Then 
\ben
\smash{\sin(\psi_S)\leq {2\Vert \mathbf{E}_L \Vert_2 }/{\beta}} \label{eqn:sin_theta}
\een
\end{lem}
\begin{proof}
Recall that $\hat{\mathbf{T}}_{ca}\!=\! \mathbf{T}_{ca} - \mf{E}_L$. To apply Wedin's theorem, we need to characterize quantities $\alpha$ and $\delta$ such that: $\sigma_S(\hat{\mathbf{T}}_{ca}) \geq \delta + \alpha, \quad \sigma_{S+1}(\mathbf{T}_{ca}) \leq \alpha$. 
From \eqref{Teigdecomp}, we have $\sigma_{S+1}(\mathbf{T}_{ca})\!=\!\sigma^2$. We choose $\alpha\!=\!\sigma^2$.
Using Weyl's ~inequality,
$\sigma_S(\mathbf{\widehat{T}}_{ca}) \geq \sigma_S(\mathbf{T}_{ca}) - \Vert \mathbf{E}_L \Vert_2\overset{(a)}{\geq}{\sigma_S(\mathbf{T}_{ca})}/{2}
\overset{(b)}{=}{(\sigma_{S}(\mf{A}_{\mathbb{U}_{\mathbb{S}}}(\blds{\theta})\mf{P}\mf{A}_{\mathbb{U}_{\mathbb{S}}}(\blds{\theta})^H)+\sigma^2)}/{2}$, where $(a)$ follows from the assumption $2\Vert \mf{E}_L\Vert_2\leq \sigma_S(\mf{T}_{ca})$ and $(b)$ follows from \eqref{Teigdecomp}. Combining with the preceding inequality, we obtain
$\sigma_S(\hat{\mathbf{T}}_{ca})-\sigma^2 \geq {(\sigma_S(\mf{A}_{\mathbb{U}_{\mathbb{S}}}(\blds{\theta})\mf{P}\mf{A}_{\mathbb{U}_{\mathbb{S}}}(\blds{\theta})^H)-\sigma^2)}/{2}
\geq \beta/{2}> 0$, 
where the last term is positive due to the given condition. Then we can choose $\delta=\sigma_S(\mf{\widehat{T}}_{\text{ca}})-\sigma^2$ which satisfies $\sigma_S(\mf{\widehat{T}}_{\text{ca}})=\alpha+\delta$ with $\delta>0$. The proof is completed by using \eqref{eqn:wedin_main}.
\end{proof}
\begin{lem}\label{lem:combo}
If $p_{\min}\sigma_S^2(\mf{A}_{\mathbb{U}_{\mathbb{S}}}(\blds{\theta}))>\sigma^2$ and 
{\footnotesize
\ben
\Vert \mathbf{E}_L \Vert_2 \leq \frac{\sigma_S(\mathbf{U}_0^{(\text{c})})(\sigma_S(\mf{A}_{\mathbb{U}_{\mathbb{S}}}(\blds{\theta})\mf{P}\mf{A}_{\mathbb{U}_{\mathbb{S}}}(\blds{\theta})^H)-\sigma^2)}{4\sqrt{2}} \label{eqn:EL_condition}
\een}
then 
$\Vert \mathbf{\Psi}^{(\text{c})} - \hat{\mathbf{\Psi}}^{(\text{c})} \Vert_2 \leq \frac{14\sqrt{2}\Vert\mathbf{E}_L \Vert_2}{\sigma_S^2(\mathbf{U}_0^{(\text{c})})(p_{\min}\sigma^2_S(\mf{A}_{\mathbb{U}_{\mathbb{S}}}(\blds{\theta}))-\sigma^2)}$.
\end{lem}
\begin{proof}
From the definition of $\mathbf{U}^{(\text{c})},\mf{\widehat{U}}^{(\text{c})}$ we have:
\begin{align}
    &\Vert \mathbf{U}^{(\text{c})}-\mf{\widehat{U}}^{(\text{c})} \Vert_2^2 =  \Vert (\mathbf{U}^{(\text{c})}-\mf{\widehat{U}}^{(\text{c})})^H(\mathbf{U}^{(\text{c})}-\mf{\widehat{U}}^{(\text{c})}) \Vert_2 \nonumber\\
    & = 2(1 - \cos(\psi_S))\leq 2(1 - \cos^2(\psi_S))=2\sin^2(\psi_S).\label{eqn:u_temp}
\end{align}
By the assumption of this lemma,
$2\Vert \mathbf{E}_L \Vert_2 \leq\sigma_S(\mathbf{U}_0^{(\text{c})})(\sigma_S(\mf{A}_{\mathbb{U}_{\mathbb{S}}}(\blds{\theta})\mf{P}\mf{A}_{\mathbb{U}_{\mathbb{S}}}(\blds{\theta})^H)\!-\!\sigma^2)$ 
and $\sigma_S(\mathbf{U}_0^{(\text{c})})\leq 1$, we have
$2\Vert \mathbf{E}_L \Vert_2\leq \sigma_S(\mathbf{T}_{ca})\sigma_S(\mathbf{U}_0^{(\text{c})})\leq \sigma_S(\mathbf{T}_{ca})
$. This together with the assumption $p_{\min}\sigma_S^2(\mf{A}_{\mathbb{U}_{\mathbb{S}}}(\blds{\theta}))>\sigma^2$ enables us to apply Lemma \ref{anglebetweensubspace}. Combining \eqref{eqn:sin_theta} with \eqref{eqn:u_temp} we obtain the following bound:
{\footnotesize
\ben
\Vert \mathbf{\widehat{U}}^{(\text{c})}-\mathbf{U}^{(\text{c})} \Vert_2 \leq \frac{2\sqrt{2}\Vert \mathbf{E}_L\Vert_2}{\sigma_S(\mf{A}_{\mathbb{U}_{\mathbb{S}}}(\blds{\theta})\mf{P}\mf{A}_{\mathbb{U}_{\mathbb{S}}}(\blds{\theta})^H)-\sigma^2}.\label{eqn:U_diff}
\een}
Notice that 
{\footnotesize
\ben
&&\|\widehat{\boldsymbol{\Psi}}^{(\text{c})}-\boldsymbol{\Psi}^{(\text{c})}\|_2=\|(\mf{\widehat{U}}_0^{{(\text{c})}\dagger}-\mf{U}_0^{{(\text{c})}\dagger})\mf{\widehat{U}}^{(\text{c})}_1\nonumber+\mf{U}_0^{^{(\text{c})}\dagger}(\mf{\widehat{U}}_1^{(\text{c})}-\mf{U}_1^{(\text{c})})\|_2 \nonumber\\&&\le \|\mf{\widehat{U}}_0^{(\text{c})\dagger}-\mf{U}_0^{(\text{c})\dagger}\|_2\|\mf{\widehat{U}}_1^{(\text{c})}\|_2+\|\mf{U}_0^{(\text{c})\dagger}\|_2\|\mf{\widehat{U}}_1^{(\text{c})}-\mf{U}_1^{(\text{c})}\|_2\nonumber\\&&\le \|\mf{\widehat{U}}_0^{(\text{c})\dagger}-\mf{U}_0^{(\text{c})\dagger}\|_2+ \|\mf{U}_0^{(\text{c})\dagger}\|_2\|\mf{\widehat{U}}^{(\text{c})}-\mf{U}^{(\text{c})}\|_2\label{eqn:psi_bound}
\een}
where the last inequality follows from the fact that 
$\mf{\widehat{U}}^{(\text{c})}_1,\mf{\widehat{U}}^{(\text{c})}_1-\mf{U}^{(\text{c})}_1$ are submatrices of $\mf{\widehat{U}}^{(\text{c})}$ and $\mf{\widehat{U}}^{(\text{c})}-\mf{U}^{(\text{c})}$, respectively. Therefore, we have $\|\mf{\widehat{U}}_1^{(\text{c})}\|_2\le \|\mf{\widehat{U}}^{(\text{c})}\|_2=1$, and $\|\mf{\widehat{U}}_1^{(\text{c})}-\mf{U}_1^{(\text{c})}\|_2\le \|\mf{\widehat{U}}^{(\text{c})}-\mf{U}^{(\text{c})}\|_2$. 
We use a result from \cite[Theorem 3.2]{hansen1987truncatedsvd} which states that a matrix $\mathbf{F}$ with rank $S$, and its perturbed matrix $\widetilde{\mathbf{F}}=\mathbf{F}+\widetilde{\mathbf{E}}$ satisfy the following inequality:
$\Vert \mathbf{F}^{\dagger}-\widetilde{\mathbf{F}}^{\dagger}\Vert_2 \leq 3\Vert \widetilde{\mathbf{E}} \Vert_2/(\sigma_S(\mathbf{F})(\sigma_S(\mathbf{F})-\Vert\widetilde{\mathbf{E}}\Vert_2))
$
provided $\Vert \widetilde{\mathbf{E}} \Vert_2 < \sigma_S(\mathbf{F})$.
From \eqref{eqn:U_diff}, and using the assumption of the lemma we have:
{\footnotesize
\begin{align}
 \Vert \widehat{\mathbf{U}}_0^{(\text{c})}-\mathbf{U}_0^{(\text{c})} \Vert_2&\leq \Vert \widehat{\mathbf{U}}^{(\text{c})}-\mathbf{U}^{(\text{c})} \Vert_2 \leq \frac{2\sqrt{2}\Vert \mathbf{E}_L \Vert_2 }{\sigma_S(\mf{A}_{\mathbb{U}_{\mathbb{S}}}(\blds{\theta})\mf{P}\mf{A}_{\mathbb{U}_{\mathbb{S}}}(\blds{\theta})^H)-\sigma^2}\nonumber\\
&\leq {\sigma_S(\mathbf{U}_0^{(\text{c})})}/{2}.\label{eqn:U0_bound_new}
\end{align}}
We can use the aforementioned result by substituting $\mathbf{F}$ with $\mf{U}_0^{(\text{c})}$, and $\widetilde{\mathbf{F}}$ with $\mf{\widehat{U}}_0^{(\text{c})}$:
$\|(\mf{\widehat{U}}_0^{(\text{c})\dagger}-\mf{U}_0^{(\text{c})\dagger})\|_2\le \frac{3\|(\mf{\widehat{U}}^{(\text{c})}_0-\mf{U}^{(\text{c})}_0)\|_2}{\sigma_S(\mf{U}^{(\text{c})}_0)(\sigma_S(\mf{U}^{(\text{c})}_0)-\|\mf{\widehat{U}}^{(\text{c})}_0-\mf{U}^{(\text{c})}_0\|_2)}\!\le\!\frac{6\|\mf{\widehat{U}}_0^{(\text{c})}-\mf{U}_0^{(\text{c})}\|_2}{\sigma_S^2(\mf{U}_0^{(\text{c})})}\!\le\!\frac{6\|\mf{\widehat{U}}^{(\text{c})}-\mf{U}^{(\text{c})}\|_2}{\sigma_S^2(\mf{U}_0^{(\text{c})})}
$, 
where the second inequality follows from \eqref{eqn:U0_bound_new}.
Combining this with \eqref{eqn:psi_bound}, we get the final bound:
$\Vert \mathbf{\Psi}^{(\text{c})} - \hat{\mathbf{\Psi}}^{(\text{c})} \Vert_2 \leq \frac{6\|(\mf{\widehat{U}}^{(\text{c})}-\mf{U}^{(\text{c})})\|_2}{\sigma_S^2(\mf{U}_0^{(\text{c})})}+\frac{1}{\sigma_{S}(\mf{U}_0^{(\text{c})})}\|(\mf{\widehat{U}}^{(\text{c})}-\mf{U}^{(\text{c})})\|_2
    \leq \frac{7\|(\mf{\widehat{U}}^{(\text{c})}-\mf{U}^{(\text{c})})\|_2}{\sigma_S^2(\mf{U}_0^{(\text{c})})}
    \leq \frac{14\sqrt{2}\Vert\mathbf{E}_L \Vert_2}{\sigma_S^2(\mathbf{U}_0^{(\text{c})})(\sigma_S(\mathbf{A}_{\mathbb{U}_{\mathbb{S}}}(\blds{\theta}))\mf{P}\mathbf{A}_{\mathbb{U}_{\mathbb{S}}}(\blds{\theta}))^H)-\sigma^2)}  \leq {14\sqrt{2}\Vert\mathbf{E}_L \Vert_2}/{\sigma_S^2(\beta\mathbf{U}_0^{(\text{c})})}.$
\end{proof}
\vspace{-0.2cm}
Next, we state the following Lemma from \cite{WeilinEsprit} that can be used to obtain a lower bound on $\sigma_S^2(\mathbf{U}_0^{(\text{c})})$.
\begin{lem}[Lemma 3, \cite{WeilinEsprit}]\label{lem:U0}
Let $\mf{U}^{(a)}$ be any orthonormal basis for $\mathcal{R}(\mathbf{A}_{\mathbb{U}_{\mathbb{S}}}(\blds{\theta}))$. Then the following holds:
$
\sigma_S^2(\mathbf{U}^{(a)}_0) \geq \max(1-\frac{S}{\sigma_S^2(\mathbf{A}_{\mathbb{U}_{\mathbb{S}}}(\blds{\theta}))},4^{-S})
$
\end{lem} 
\vspace{-0.5cm}
\begin{proof}[Proof of \cref{Espiriterror}]
Define $C_S=\frac{2^{-S}}{4\sqrt{2}}$ and $C_S'= 14\pi\sqrt{2}S^{3/2}4^{S}$. Under Assumption \textbf{A3}, these quantities are constants since $S$ is held fixed. If $\beta>0$ and the assumption $\Vert \EL \Vert_2 \leq C_S\beta$ ensures that condition \eqref{eqn:EL_condition} holds since $\sigma_S(\mathbf{U}_0^{(\text{c})})\geq 2^{-S}$ from \cref{lem:U0}. Now, we can apply \cref{lem:combo} to bound $\Vert\mathbf{\Psi}^{(\text{c})} - \hat{\mathbf{\Psi}}^{(\text{c})} \Vert_2$. We plug this bound on $\Vert\mathbf{\Psi}^{(\text{c})} - \hat{\mathbf{\Psi}}^{(\text{c})} \Vert_2$ in \eqref{eqn:md_psi_bound}:
$\text{md}(\blds{\theta},\blds{\widehat{\theta}}) \leq 14\sqrt{2}\pi \frac{S^{3/2}q\Vert\mathbf{E}_L \Vert_2}{\sigma_S^2(\mathbf{U}_0^{(\text{c})})C_S'}\leq q \Vert\mathbf{E}_L \Vert_2$, 
where the last inequality follows from the bound $\sigma^2_S(\mathbf{U}_0^{(\text{c})})\geq 4^{-S}$ in \cref{lem:U0}.
\end{proof}
\vspace{-0.5cm}
\subsection{Proof of Theorem 3}\label{app:prob_coesp}
We will utilize \cref{thm:coverr} and \cref{Espiriterror} to prove \cref{thm:main_coesp}. One can see from Theorem \ref{Espiriterror} that under the assumptions $\beta>0$ and $\Vert \mathbf{E}_L \Vert_2 \leq C_S\beta$ we can bound $\text{md}(\blds{\theta},\widehat{\blds{\theta}}) \leq q\Vert \mathbf{E}_L \Vert_2$.
For a given $\epsilon>0$, two cases arise:

\textbf{Case I ($\epsilon \leq C_S\beta q$):}
In this case, $\text{min}(C_S\beta,\frac{\epsilon}{q})=\epsilon/q$. Therefore, $\Vert \mathbf{E}_L \Vert_2 \leq \frac{\epsilon}{q} \Rightarrow \Vert \mathbf{E}_L \Vert_2 \leq C_S\beta $, and from \cref{Espiriterror} the matching distance error is less than $\text{md}(\blds{\theta},\widehat{\blds{\theta}}) \leq \epsilon$. This means
$P(\text{md}(\blds{\theta},\widehat{\blds{\theta}}) \leq \epsilon) \geq P\big(\Vert \mathbf{E}_L \Vert_2 \leq \frac{\epsilon}{q}\big)$.
From \cref{thm:coverr}, we can obtain the following tail bound:
{\footnotesize
\begin{eqnarray}  
&\hspace{-2in}P(\Vert \mathbf{E}_L \Vert_2 \leq  \nonumber
\frac{\epsilon}{q}) \geq \\ & 1 -8 M_{ca} \exp \left[ -c_1 L \min \left( \frac{c_2\epsilon^2 }{q^2\| \mathbf{R_{y}} \|_2^2 \Delta(\mathbb{S})}, \frac{\epsilon }{q\| \mathbf{R_{y}} \|_2 \sqrt{\Delta(\mathbb{S})}}\right) \right]. \label{eqn:case_1}
\end{eqnarray}}
\textbf{Case II ($\epsilon > C_S\beta q$):}
For values of $\epsilon$ satisfying $\epsilon >C_s\beta q$, we have $\min(C_S\beta,\epsilon/q)=C_S\beta$. Therefore, if $\Vert \mathbf{E}_L \Vert_2\leq C_S \beta$, then from \cref{Espiriterror} we have $\text{md}(\blds{\theta},\widehat{\blds{\theta}}) \leq C_S\beta q$. We obtain the following bound on the tail probability due to \cref{thm:coverr}, 
{\footnotesize
\begin{eqnarray}  
&\hspace{-2in}P(\Vert \mathbf{E}_L \Vert_2 \leq  \nonumber
C_S\beta ) \geq \\ & 1 -8 M_{ca} \exp \left[ -c_1 L \min \left( \frac{c_2C_S^2\beta^2 }{\| \mathbf{R_{y}} \|_2^2 \Delta(\mathbb{S})}, \frac{C_S\beta }{\| \mathbf{R_{y}} \|_2 \sqrt{\Delta(\mathbb{S})}}\right) \right]. \label{eqn:case_2}
\end{eqnarray}}
If the number of snapshots $L$ satisfy the following bound:
$   L \geq c_3\ln\big(\frac{8M_{ca}}{\delta}\big) \max \big( \frac{q^2_1 \Delta(\mathbb{S})}{c_2\epsilon^2}, \frac{q_1 \sqrt{\Delta(\mathbb{S})}}{\epsilon},\frac{L_0^2}{c_2},L_0\big)$, 
where $q_1=q\| \mathbf{R_{y}} \|_2, c_3=1/c_1$ and $L_0= \frac{ \| \mathbf{R_{y}} \|_2 \sqrt{\Delta(\mathbb{S})}}{C_S\beta}$
then combining \eqref{eqn:case_1} and \eqref{eqn:case_2} we obtain the following bound $P\big(\text{md}(\blds{\theta},\widehat{\blds{\theta}}) \leq \min(\epsilon, C_S\beta q))\geq 1-\delta$.

\end{appendices}
\vspace{-0.5cm}
\bibliographystyle{IEEEtran}
\bibliography{ref} 
\end{document}